\documentclass[11pt,a4paper,reqno]{amsart} 

\usepackage{hyperref}

\usepackage{amsmath}
\usepackage{amsthm}
\usepackage{amsfonts}
\usepackage{amssymb}
\usepackage{array}
\usepackage[margin=2.7cm]{geometry} 

\usepackage{enumitem}

\newcommand{\cal}{\mathcal}

\newtheorem{theorem}{Theorem}[section]

\newtheorem{proposition}[theorem]{Proposition}
\newtheorem{corollary}[theorem]{Corollary}
\newtheorem{lemma}[theorem]{Lemma}

\newtheorem{problem}{Problem}

\theoremstyle{definition}
\newtheorem{definition}[theorem]{Definition}
\newtheorem{example}[theorem]{Example}
\newtheorem{remark}[theorem]{Remark}

\newcommand{\C}{\mathcal{C}}
\newcommand{\A}{\mathcal{A}}

\newcommand{\cs}{\mathrm{colsp}}
\newcommand{\GL}{\mathrm{GL}}
\newcommand{\F}{\mathbb{F}}
\newcommand{\fq}{\mathbb{F}_q}

\newcommand{\N}{\mathbb{N}}
\newcommand{\rk}{\mathrm{rk}}

\newcommand{\trk}{\mathrm{trk}}
\newcommand{\rs}{\mathrm{rowsp}}

\newcommand{\csup}{\mathrm{csupp}}
\newcommand{\rsup}{\mathrm{rsupp}}
\newcommand{\mU}{\mathcal{U}}
\newcommand{\mC}{\mathcal{C}}
\newcommand{\mat}{\F_q^{n \times m}}
\newcommand{\drk}{d}

\newcommand{\Fm}{\mathbb{F}_{q^m}}

\newcommand{\Fq}{\mathbb{F}_{q}}
\newcommand{\FN}{\F^{k \times n \times m}}
\renewcommand{\ss}{\mathrm{ss}}
\renewcommand{\trk}{\mathrm{trk}}
\newcommand{\diag}{\mathrm{diag}}

\renewcommand{\longrightarrow}{\to}

\newcommand{\Fqk}{\Fq\text{-}[n\times m, k]}
\newcommand{\Fqkd}{\Fq\text{-}[n\times m, k, d]}

\newcommand{\npmatrix}[1]{\left( \begin{matrix} #1 \end{matrix} \right)}

\usepackage[dvipsnames]{xcolor}
\usepackage{pagecolor}
\definecolor{light-gray}{gray}{0.90}

\usepackage{tikz}
\usepackage{pgfplots}

\title{Tensor Representation of Rank-Metric Codes}

\author[Eimear Byrne]{Eimear Byrne}
\address{School of Mathematics and Statistics, University College Dublin, Belfield, Ireland}
\curraddr{}
\email{ebyrne@ucd.ie}
\thanks{}

\author[Alessandro Neri]{Alessandro Neri$^*$}
\address{Institute of Mathematics, University of Zurich, Switzerland}
\curraddr{}
\email{alessandro.neri@math.uzh.ch}
\thanks{}

\author[Alberto Ravagnani]{Alberto Ravagnani$^*$}
\address{School of Mathematics and Statistics, University College Dublin, Belfield, Ireland}
\curraddr{}
\email{alberto.ravagnani@ucd.ie}
\thanks{$^*$Alessandro Neri was supported by the Swiss National Science Foundation through grant no. 169510. Alberto Ravagnani was supported by the Swiss National Science Foundation through grant no. P2NEP2\_168527 and by the Marie Curie Research Grants
Scheme, grant no. 740880.}

\author[John Sheekey]{John Sheekey}
\address{School of Mathematics and Statistics, University College Dublin, Belfield, Ireland}
\curraddr{}
\email{john.sheekey@ucd.ie}
\thanks{}


\begin{document}


\maketitle
\thispagestyle{empty}

\begin{abstract}
    We present the theory of rank-metric codes with respect to the 3-tensors that generate them. We define the generator tensor and the parity check tensor of a matrix code, and describe the properties of a code through these objects. 
    We define the tensor rank of a code to be the tensor rank of its generating tensors, and propose that this quantity is a significant coding theoretic parameter. By a result on the tensor rank of Kruskal from the 1970s, the tensor rank of a rank-metric code of dimension $k$ and minimum rank distance~$d$ is at least $k+d-1$. We call codes that meet this bound \textit{minimal tensor rank} (MTR) codes.
    It is known from results in algebraic complexity theory that an MTR code implies the existence of an MDS code. In this paper, we also address the converse problem, that of the existence of an MTR code, given an MDS code.
    We  identify several parameters for which the converse holds and give explicit constructions of MTR codes using MDS codes.  
    We furthermore define generalized tensor ranks, which give a refinement of the tensor rank as a code invariant.
Moreover, we use these to distinguish inequivalent rank-metric codes.    
\end{abstract}

\bigskip
  
\bigskip

\section{Introduction}

The theory of rank-metric codes is an important topic in coding theory, which has seen a resurgence of interest in the last decade.
Any linear space of matrices can be viewed as a rank-metric code, where the rank distance between a pair of matrices is the rank of their difference and indeed the most general class of (linear) rank-metric codes
are the $\fq$-subspaces of $\fq^{n\times m}$, which we refer to as {\em matrix codes}.
Much of the initial focus of the theory was on subspaces of $\F_{q^m}^n$, the so-called $\F_{q^m}$-linear {\em vector rank-metric codes},
where the rank of a vector in $\F_{q^m}^n$ is the dimension of the span of its coordinates as an $\F_q$-linear space. 
Given a basis of $\F_{q^m}$ over $\fq$, any vector rank-metric code can be represented as a linear space of matrices in $\fq^{n \times m}$.
The most prominent family in this class of codes are the Delsarte-Gabidulin codes \cite{del,gabid,roth}, which can be conveniently described 
in terms of generator and parity check matrices that are $q$-analogues of those for the well-studied Reed-Solomon codes of classical coding theory.

The role of matrices in classical (Hamming metric) coding theory is crucial. Efficient encoding and decoding rely on generator and parity check matrices. Several properties of a code are characterized by such matrices, including duality, equivalence, and minimum distance.
These matrices also yield connections between coding theory and finite geometry, from which
optimal codes have been constructed from sets of points in projective space.

In this paper, we present rank-metric codes in the framework of 3-tensors. More precisely, we define the generator tensor and parity check tensor of 
an $\fq$-linear space of matrices and describe the properties of such codes in relation to these objects.

An important and well-studied parameter of a tensor is given by its {\em tensor rank}. This aspect of bilinear forms is central to algebraic complexity theory \cite{brockett,algcplex,kruskal}.
The definition of tensor rank considered here is the minimum number of simple tensors that appear in the expression of a tensor as a sum of simple tensors.
It extends the notion of matrix rank and gives a measure of the complexity of tensor multiplication.  
Precise computation of tensor rank is elusive for an arbitrary tensor; indeed computing the rank of a 3-tensor over a finite field is NP-complete~\cite{np}.
We propose that tensor rank is a significant parameter in the theory of rank-metric codes. This extends the notion of the tensor rank of a rank-metric code corresponding to a finite semifield \cite{Lavrauw}. A rank-metric code in $\F_q^{n \times m}$ is a {\em slice space} of
an associated {\em generator tensor}, just as a code in $\fq^n$ is the row-space of a generator matrix. The smaller the tensor rank of the generating tensor, the more efficient the encoding. It is therefore of interest to obtain codes whose generating tensors have minimum tensor rank. 

Lower bounds on tensor rank have been known for some time \cite{kruskal}. If $X$ is a generating tensor for an $\F_q$-linear code in $\Fq^{n\times m}$ of dimension $k$ and minimum rank distance $d$ then this lower bound on $\trk(X)$, the tensor rank of $X$, can be expressed as:
\begin{eqnarray}\label{eq:intro}
\trk(X) \geq k + d - 1.
\end{eqnarray} 
Coding theorists will immediately notice the similarity of this inequality to the Singleton bound. We will refer to a code having a generating tensor meeting this bound as a {\em minimum tensor rank} (MTR) code.
It is known that any (nondegenerate) tensor of rank $R$ gives rise to a linear block code of length $R$, and in particular that any lower bound on the length of a linear block code provides an immediate lower bound on the tensor rank \cite{brockett,algcplex}. It can therefore be deduced that any MTR code gives a construction of an MDS block code. A central problem posed in this paper is to address the inverse problem: given an MDS block code of length $R$, find a construction of an MTR code with tensor rank $R$. We solve this problem for a range of parameter sets.
  
  We introduce the \textit{generalized  ranks} of a matrix code, which turn out to be an invariant of code equivalence. In particular, such values can be used to distinguish between inequivalent codes and, remarkably, even between MRD codes that otherwise share many invariants. 
Moreover, generalized tensor ranks lead to a refinement of the tensor rank bound, from which the existing tensor rank bound (\ref{eq:intro}) can be deduced. The coding theoretic arguments used in these proofs are very simple and compact. 

A further aspect of the tensor description of a matrix code is that many of its coding theoretic parameters are encoded in its generating and parity check tensors. For example, the minimum rank distance of a matrix code can be characterized by the dimensions of its slice spaces, in direct analogy with relation of the minimum Hamming distance of a block code in relation to its parity check matrix.  

\subsection*{Outline.} In Section~\ref{sec:2} we cover preliminary results on rank-metric codes and in Section~\ref{sec:3} we recall basic results on 3-tensors and characterize tensor rank. 
In Section~\ref{sec:4} we relate matrix codes and tensors and introduce the generator tensor of a code. We define optimality of a code with respect to tensor rank and describe a connection between matrix codes of a fixed tensor rank and linear block codes, which allows the construction of one code, given the other. In Section~\ref{sec:5} we define the notion of an {\em extremal triple} $(C,V,W)$. It  consists of an optimal $\fq$-code $C$, which is a code with smallest possible length for given dimension and minimum distance,  and full-rank matrices $V \in \fq^{n \times R},W\in\fq^{m \times R} $; such a triple yields a \emph{tensor rank extremal  code}. In the special case where $C$ is an MDS code, the extremal triple yields an MTR code. We identify parameters $m,n,R$ for which $(C,V,W)$ is always an extremal triple via a characterization theorem. For some values outside these parameter sets, we give an explicit construction of an extremal triple using Cauchy codes. 
We furthermore consider the Delsarte-Gabidulin codes and obtain an upper bound on their tensor rank. We apply these results to  give constructions of matrix codes with upper-bounded tensor rank.  
In Section~\ref{sec:6} we define the $r$-th generalized tensor rank of a matrix code and establish their main properties. We show that the $r$-th generalized tensor rank is an invariant of equivalent matrix codes and use this to distinguish between codes. We also show that these ranks are not an invariant of duality. Finally, in Section~\ref{sec:7}, we define the parity check tensor of a matrix code. 
We consider the standard coding theoretic operations of shortening and puncturing of matrix codes, and use these as a tool to relate the parameters of a code to its generator and parity check tensors.

\section{Rank-Metric Codes}
\label{sec:2}

Throughout this paper, $q$ is a prime power and $m,n$ are positive integers. In this section, we assume $m \ge n$ to simplify the presentation. Analogous results hold for $n \le m$.
We denote by $\mat$ the $\F_q$-linear space of $n \times m$ matrices with entries in $\F_q$. For an integer $i \ge 0$, we let $[i]:=\{1,\ldots,i\}$. All dimensions are computed over~$\F_q$, unless otherwise stated.

The main objects studied in this paper are rank-metric codes. They were 
introduced in by Delsarte \cite{del} for combinatorial interest.

\begin{definition}
The \textbf{rank distance} between $X,Y \in \mat$ is $\drk(X,Y):=\rk(X-Y)$. A (\textbf{rank-metric}) \textbf{code} is an $\F_q$-linear subspace $\mC \subseteq \mat$. If $\mC \neq \{0\}$, then the \textbf{minimum distance} of $\mC$ is the integer 
$$\drk(\mC):= \min\{\rk(X) : X \in \mC, \ X \neq 0\} = \min\{\drk(X,Y) : X,Y \in \mC, \ X \neq Y\}.$$
\end{definition}

From now on, we will refer to a rank-metric code $\C \subseteq \Fq^{n\times m}$ of dimension $k$  as an $\Fqk$ code. When the minimum distance $d$ is known, we will call it an $\Fqkd$ code.

In this paper we adopt the following definition of code equivalence.

\begin{definition}
Codes $\mC,\mC' \subseteq \mat$ are \textbf{equivalent} if there exists an $\F_q$-linear isometry $\varphi: (\mat,d) \to (\mat,d)$ such that $\varphi(\mC)=\mC'$.
\end{definition}

As a linear isometry of $\mat$ is necessarily bijective, equivalent codes have the same dimension and minimum distance. According to \cite{classif,wan}, in which all the $\Fq$-linear isometries are classified,  
 codes $\C,\C'\subseteq \mat$ are equivalent if and only if there exist invertible matrices $A \in \GL(n,q)$, $B\in \GL(m,q)$ such that 
$$\C'=A\C B:=\left\{AXB : X \in \C\right\},$$
or, when $m=n$,
$$\C'=A\C^\top B:=\left\{AX^\top B : X \in \C\right\}.$$

The following result is the rank-metric analogue of the Singleton bound for codes with the Hamming metric.

\begin{theorem}[Theorem 5.4 of \cite{del}] \label{singbound}
Let $\mC \subseteq \mat$ be a non-zero code. Then $$\dim(\mC) \le m(n-\drk(\mC)+1).$$
\end{theorem}

\begin{definition}
A code $\mC$ is  \textbf{maximum rank distance} (\textbf{MRD}) if it meets the bound of Theorem \ref{singbound}, or if it is the zero code.
\end{definition}

Recall that the \textbf{trace-product} of $X,Y \in \mat$ is $\langle X,Y \rangle:=\mbox{Tr}(XY^\top)$. It is well-known and easy to see that the map
$(X,Y) \mapsto \mbox{Tr}(XY^\top)$ defines a bilinear, symmetric and nondegenerate form on $\mat$.

\begin{definition}
The \textbf{dual} of an $\Fqk$ code is $$\mC^\perp:=\{X \in \mat : \langle X,Y \rangle =0 \mbox{ for all } Y \in \mC\}.$$
Note that $\mC^\perp$ is an $\Fq[n\times m, nm-k]$ code.
\end{definition}

In \cite{gabid}, Gabidulin introduces a class of rank-metric codes that are linear over the extension field $\F_{q^m}$. They are defined as follows.

\begin{definition}
A \textbf{vector rank-metric code} is an $\F_{q^m}$-subspace $C \subseteq \F_{q^m}^n$. 
\end{definition}

To obtain a matrix code from a vector code, it suffices to use that fact that
$\F_{q^m}^n$ and $\mat$ are isomorphic as $\F_q$-linear spaces. An isomorphism can be constructed as follows. Let $\Gamma=\{\gamma_1,\ldots,\gamma_m\}$ be a basis of $\F_{q^m}/\F_q$. For $v \in \F_{q^m}$, denote by $\Gamma(v) \in \mat$ the matrix whose $(i,j)$ entry is the $j$-th coordinate of $v_i$ over the basis $\Gamma$. Then the map $v \mapsto \Gamma(v)$ is an $\F_q$-isomorphism. We denote by
$\Gamma(C)$ the image of a vector rank-metric code $C \subseteq \F_{q^m}^n$ under~$\Gamma$, i.e., we let
$\Gamma(C)=\left\{\Gamma(v) : v \in C  \right\}$.

\begin{lemma}[see e.g. \cite{costch}]
Let $C \subseteq \F_{q^m}^n$ be a non-zero vector code. The minimum distance of $\Gamma(C)$ does not depend on the choice of the basis
$\Gamma$ for $\F_{q^m}/\F_q$. Moreover, for any such basis we have
$$\dim_{\F_q}(\Gamma(C)) = m \cdot \dim_{\F_{q^m}}(C).$$
\end{lemma}

\begin{definition}
The \textbf{minimum distance} of a non-zero vector code $C \subseteq \F_{q^m}^n$ is the minimum distance of $\Gamma(C)$, where $\Gamma$ is any basis of
$\F_{q^m}/\F_q$. It is denoted by $\drk(C)$.
\end{definition}

With these definitions, it is easy to see that a vector code $C\subseteq \Fm^n$ is MRD if and only if
$$\drk(C)=n-\dim_{\Fm}(C)+1.$$
MRD vector codes (and therefore MRD matrix codes) exist for every set of parameters. The first construction was found by Delsarte \cite{del} and independently by Gabidulin \cite{gabid} and Roth \cite{roth}. It was then
generalized in~\cite{gabidgen}. 

Let $K,s,m$ be positive integers such that $1\leq K \leq m$ and $1\leq s <m$. Define the set $\mathcal G_{K,s}$ of linearized polynomials as
$$\mathcal G_{K,s}:=\left\{ \sum_{i=0}^{K-1}f_ix^{q^{si}} : f_i \in \Fm \mbox{ for } i=0,\ldots, K-1 \right\}.$$

\begin{definition}\label{def:Gabcode}
Let $\alpha=(\alpha_1,\ldots,\alpha_n)\in \Fm^n$ be a vector such that $\alpha_1,\ldots, \alpha_n \in \Fm$ are linearly independent over $\Fq$. Let moreover $K,n,s,m$ be positive integers such that
$1\leq K \leq n \leq m$, $1\leq s <m$ and $\gcd(s,m)=1$. The \textbf{generalized Delsarte-Gabidulin code} $\mathcal G_{K,s}(\alpha)$ is defined as
$$\mathcal G_{K,s}(\alpha):=\left\{(f(\alpha_1),\ldots,f(\alpha_n)) : f \in \mathcal G_{K,s} \right\}.$$
When $s=1$, we will simply refer to $\mathcal G_{K,1}(\alpha)$ as a \textbf{Delsarte-Gabidulin code}.
\end{definition}


Another property on vector codes that we will need is the following, that is different matrix representations of the same vector code lead to equivalent rank-metric codes.

\begin{remark}\label{rem:equivvector}
Let $C \subseteq \Fm^n$ be a vector code, and let $\Gamma=\{\gamma_1,\ldots, \gamma_m\}$, $\Gamma'=\{\gamma_1',\ldots, \gamma_m'\}$ be bases of $\F_{q^m}/\F_q$. Note that the matrix codes $\Gamma(C)$ and $\Gamma'(C)$ are equivalent. 
\end{remark}


We also define the column support and row support of a rank-metric code.
See~\cite{gorla2019} for a detailed analysis of the various definitions of rank-support proposed in the literature.

\begin{definition}
	Let $\C$ be a rank-metric code. The {\bf column support} and the {\bf row support} of $\C$ are defined to be the $\Fq$-subspaces of $\Fq^n$ and $\Fq^m$, respectively,
	defined by 
	$$\csup(\C):=\sum_{M \in \C} \mathrm{colsp}(M), \qquad \rsup(\C):=\sum_{M\in \C} \mathrm{rowsp}(M),$$
	where the sums are sums of vector subspaces. The code $\C\subseteq \Fq^{n\times m}$ is said to be {\bf nondegenerate} if $\csup(\C)=\Fq^n$ and $\rsup(\C)=\Fq^m$.
\end{definition}



\section{3-Tensors}
\label{sec:3}

We recall some definitions and results from tensor algebra. The interested reader is referred to \cite{algcplex,advlinalg} for more details. 
In this section, $\F$ denotes an arbitrary field.

Recall that a tensor product of $\F$-spaces $U$ and $V$, denote by $U \otimes V$, is defined as a pair $(T,\varphi)$, where $\varphi: U \times V \to T$ is a bilinear map to the $\F$-space $T$ such that
for any bilinear map $f: U \times V \to W$ to an $\F$-space $W$, there exists a unique $\F$-linear map $\hat{f}: T \longrightarrow W$ satisfying $f=\hat{f} \circ \varphi$. We say that 
$(T, \varphi)$ satisfies the {\em universal mapping property}. The existence and uniqueness of $(T,\varphi)$, and hence the well-definedness of $U \otimes V$, can be shown by its construction 
as a quotient space of the free $\F$-linear space on $U \times V$ (see, for example, \cite[Chapter~10]{advlinalg}). 
Tensors of the form $u \otimes v$ 
are called {\bf simple} tensors (also called {\bf fundamental} or {\bf pure} tensors in the literature). Arbitrary elements of $U \otimes V$ are expressed as sums of simple tensors: $\sum_{i=1}^\ell u_i \otimes v_i$, with $u_i \in U$ and $v_i \in V$.
Since the tensor product of a pair of spaces is itself a vector space, we may construct the tensor product $(U \otimes V) \otimes W = U \otimes (V \otimes W)$, for $\F$-spaces $U,V,W$, which we therefore express as 
$U \otimes V \otimes W$. The corresponding map associated with such a tensor product is a trilinear map $\varphi: U \times V \times W \longrightarrow U \otimes V \otimes W$. 

If $\{u_1,\ldots,u_k\}$, $\{v_1,\ldots,v_n\}$ and $\{w_1,\ldots,w_m\}$ are bases of $U$, $V$ and $W$, respectively, then a basis of $U \otimes V 
\otimes W$ is given by
$$\{u_i \otimes v_j \otimes w_\ell : 1 \le i \le k, \ 1 \le j \le n, \ 1 \le \ell \le m\}.$$ 
In particular,  $\dim_\F(U \otimes V \otimes W) = \dim_\F(U)  \dim_\F(V) \dim_\F(W).$

In this paper we shall be mainly interested in tensor products of the form
\begin{equation*} \label{stan}
\F^{k} \otimes \F^{n} \otimes  \F^{m},
\end{equation*}
whose elements are called $3$-tensors, 3rd-order tensors, or triads.
The elements of this space can be represented as 3-dimensional arrays. As with matrices (2nd-order tensors), one can define a 3-dimensional array of size $k \times n \times m$ as a function 
$$X: \{1,\ldots,k\} \times \{1,\ldots,n\} \times \{1,\ldots,m\} \to \F,$$
which we represent as
\begin{equation*}\label{eq:1}
	X=(X_{ij\ell} : 1 \leq i\leq k,  \ 1 \leq j\leq n, \ 1\leq \ell \leq m).
\end{equation*}
These representations of the tensor $X=\sum_{r=1}^R u_r\otimes v_r \otimes w_r$ are related by
$$X_{ij\ell}=\sum_{r=1}^Ru_{ir}v_{jr}w_{\ell r},$$
where $u_r=(u_{ir}: 1 \leq i \leq k)$, $v_r=(v_{jr}: 1 \leq j \leq n)$, and $w_r=(w_{\ell r}: 1 \leq \ell \leq m)$.
We hence identify $\F^{k} \otimes \F^{n} \otimes  \F^{m}$ with the space $\FN$. The representation of $X$ as an element of $\FN$ is called its \textbf{coordinate tensor}.

For the remainder of the paper, given vectors $z_r \in \F^N$, we will write $z_{jr}$ to denote the $j$th coefficient of $z_r$ for each $r$. 
That is, $z_r:=(z_{jr}: 1 \leq j \leq N)$.





We introduce the following maps, which defines multiplication of 3-tensors with vectors (corresponding to $s=1$) and matrices ($s>1$).
$$m_1: \F^{s\times k} \times  \FN  \longrightarrow  \F^{s\times n \times m}: (A,X) \mapsto m_1(A, X)=\sum_i (Au_i)\otimes v_i \otimes w_i,$$

$$m_2: \F^{s\times n} \times  \FN  \longrightarrow  \F^{s\times k \times m}: (B,X) \mapsto m_2(B, X)=\sum_i u_i\otimes (Bv_i) \otimes w_i,$$

$$m_3: \F^{s\times m} \times  \FN  \longrightarrow  \F^{s\times k \times n}: (C,X) \mapsto m_3(C, X)=\sum_i u_i\otimes v_i \otimes (Cw_i),$$
for any $X =\sum_i u_i\otimes v_i \otimes w_i \in \F^{k \times n \times m}$.

	Let $X \in \F^{N_1 \times N_2 \times N_3}$. For each $i \in \{1,2,3\}$ and for any $A \in \F^{s \times \ell}, \ B \in \F^{\ell \times N_i}$ it is easy to see that
	\begin{equation}\label{eq:assoc}
	     m_i(AB,X)=m_i(A,m_i(B ,X)).
	\end{equation}
	Indeed, $\GL(N_i,q)$ acts on the set of tensors $\F^{N_1 \times N_2 \times N_3}$.

\begin{remark}
	Notice that, in the case that $s=1$, the operation $m_1$ yields a 3-tensor of the form 
	$\sum_i \lambda_i\otimes v_i \otimes w_i,$ 
	for some scalars $\lambda_i \in \F$, which can be identified with the 2-tensor $\sum_i (\lambda_i v_i) \otimes w_i \in \F^{n \times m}$ ($\F\otimes V$ and $V$ are isomorphic).
	Similarly, $m_2$ and $m_3$ yield 2-tensors for the case $s=1$.
	With abuse of notation, in this case we will consider the images of the $m_i$ to be in the space of matrices over $\F$.
\end{remark}

\begin{definition}
	Let $X\in \F^{N_1 \times N_2 \times N_3}$. For each $i \in \{1,2,3\}$, we define the $i$-th {\bf slice space} of $X$ to be the $\F$-span
	of $\{m_i(e_j,X) : 1 \leq j \leq N_i\}$, that is, 
	$$\mathrm{ss}_i(X):=\langle m_i(e_1,X), \ldots, m_i(e_{N_i},X)\rangle.$$ We write $\dim_i(X)$ to denote the 
	dimension of $\mathrm{ss}_i(X)$ as an $\F$-vector space.
	We say that $\ss_i(X)$ is {\bf nondegenerate} if $\dim_i(X)=N_i$, in which case we say that $X$ is {\bf $i$-nondegenerate}.
\end{definition}
If $X = \sum_{r=1}^R u_r \otimes v_r \otimes w_r \in \FN$, then clearly 
$$\ss_1(X) = \left\langle  \sum_{r=1}^R u_{jr}v_r \otimes w_r : 1 \leq j \leq k  \right\rangle, $$
where for each $r$, $u_r = (u_{jr}:1 \leq j \leq k ) \in \F^k$. In particular, $\ss_1(X)$ is the $\F$-span
of $k$ matrices 
$$A_j = \sum_r u_{jr}v_r \otimes w_r = m_1(e_j,X) \in \F^{n \times m},$$ of rank at most $R$, which form a basis of $\ss_1(X)$ if $X$ is $1$-nondegenerate.

We also point out the simple fact that for every basis $g_1,\ldots, g_{N_i}$ of $\F^{N_i}$
we have  
$$\ss_i(X)=\langle m_i(g_1,X), \ldots, m_i(g_{N_i},X) \rangle.$$
In particular, for every $G\in \GL(N_i,q)$ we have
$$\ss_i(X)=\ss_i(m_i(G,X)).$$

Of particular interest in this paper, is the 1st slice space $\ss_1(X)$ of a nondegenerate $3$-tensor $X \in \F_q^{k \times n \times m}$, which will be a $k$-dimensional subspace of matrices in $\F_q^{n \times m}$.

A notable parameter of a tensor that relates to algebraic complexity is its {\em tensor rank}, which we now define. 

\begin{definition}
	Let $X \in \FN$. The {\bf tensor rank} of $X$ is the minimum integer $R$ such that there exists $u_r \in \F^{k}, v_r \in \F^{n}, w_r \in \F^{m}$ such that 
	$$X = \sum_{r=1}^R u_r \otimes v_r \otimes w_r.$$
	We write $\mathrm{trk}(X)$ to denote the tensor rank of $X$.
	A representation of the $X$ as sum of $R=\mathrm{trk}(X)$ simple tensors is called a {\bf minimal rank form} of $X$.
\end{definition}

The reader will easily verify that 
\begin{equation}\label{eq:trkmi}
     \trk(m_1(A,X)) \leq \trk(X),
\end{equation} 
for any $A \in \F^{s \times k}$ (with analogous statements for elements in the image of $m_2$ and $m_3$). Moreover, it is straightforward to check that the tensor rank is invariant under any permutation of the spaces $\F^k,\F^n,\F^m$.

The following result gives various characterizations of the tensor rank; see for example \cite[Proposition~14.45]{algcplex}. As we will use the construction of these characterizations in Lemma \ref{lem:phipsi}, we include a proof. 

\begin{proposition}\label{prop:trkcharact}
	Let $X \in \FN$ and let $R>0$ be an integer. The following are equivalent.
	\begin{enumerate}
		\item $\trk(X)\leq R$.
		\item There exist $A_1,\ldots, A_R\in \F^{n \times m}$ of rank 1 such that $\ss_1(X)\subseteq \langle A_1,\ldots, A_R\rangle$.
		\item \label{pthree} There exist diagonal matrices $D_1,\ldots, D_{k} \in \F^{R \times R}$, and matrices $P \in \F^{n\times R}$, $Q \in \F^{m\times R}$ such that 
		$$\ss_1(X)= P\langle D_1,\ldots, D_{k}\rangle Q^\top:=\langle PD_1Q^\top, \ldots, PD_{k}Q^\top\rangle.$$
	\end{enumerate}
\end{proposition}

\begin{proof}
	Suppose that $\trk(X)\leq R$. Then $X = \sum_{r=1}^R u_r \otimes v_r \otimes w_r$ for some vectors $u_r \in \F^k$, $v_r \in \F^n$, $w_r \in \F^m$ and 
	$$\ss_1(X) = \left\langle  \sum_r u_{jr}v_r \otimes w_r : 1 \leq j \leq k  \right\rangle \subseteq \langle v_r \otimes w_r : 1 \leq r \leq R \rangle.  $$
	Conversely, if $\ss_1(X)$ is contained in the span of $R$ rank 1 matrices $A_r=v_r \otimes w_r$, then for all $1 \leq j \leq k$ there exist $u_{jr} \in \F$ satisfying $m_1(e_j,X) = \sum_{r=1}^R u_{jr} v_r \otimes w_r$. Therefore $X =\sum_{r=1}^R u_r \otimes v_r \otimes w_r$ and $\trk(X)\leq R$.
	
	Again suppose that $X = \sum_{r=1}^R u_r \otimes v_r \otimes w_r$ for some $u_r \in \F^k,v_r \in \F^n,w_r \in \F^m$. 
	Let $$D_j=\text{ diag}(u_{jr},1 \leq r \leq R).$$ So the diagonal elements of $D_j$ are the $j$th coefficents of the $u_r$. Set 
	$$P=(v_{jr}: 1\leq j \leq n, 1 \leq r \leq R), \qquad  Q=(w_{jr}: 1\leq j \leq m, 1 \leq r \leq R).$$ 
	Then $PD_jQ^\top = \sum_{r=1}^R u_{jr} v_r \otimes w_r$ for each $j$ and hence $\ss_1(X) = P\langle D_1,\ldots, D_{k}\rangle Q^\top$. 
	Conversely, given the existence of matrices $P,Q,D_i$ satisfying (\ref{pthree}), the tensor $X$ can be constructed as $\sum_{r=1}^R u_r \otimes v_r \otimes w_r$,
	where $v_r$ is the $r$th column of $P$, $w_r$ is the $r$th column of $Q$ and $u_{jr}$ is the $r$th element of the main diagonal of $D_j$ for each $j$. 
\end{proof}



\begin{example}
The following example, adapted from \cite{LaSh233}, illustrates the preceding definitions and propositions. Consider the tensor
$X =  e_1\otimes (e_1 \otimes e_1 + e_2 \otimes e_2)+ 
e_2\otimes (e_1 \otimes e_2 + e_2 \otimes e_3)$ in $\F^2\otimes\F^2\otimes \F^3$, where $\F$ is any field of characteristic not two. Then
\[
\ss_1(X)= \left\langle \npmatrix{1&0&0\\0&1&0},\npmatrix{0&1&0\\0&0&1}\right\rangle.
\]
Then $X$ can be written as the sum of the three rank one tensors 
\begin{align*}
X_1&=e_1\otimes e_1\otimes(e_1-e_3)\\
X_2&= \frac{1}{2}(e_1+e_2)\otimes (e_1+e_2)\otimes(e_2+ e_3)\\
X_3&= \frac{1}{2}(-e_1+e_2)\otimes (e_1-e_2)\otimes(e_2- e_3),
\end{align*}
which corresponds to the fact that $\ss_1(X)$ is contained in 
\[
\left\langle \npmatrix{1&0&-1\\0&0&0},\npmatrix{0&1&1\\0&1&1},\npmatrix{0&1&-1\\0&-1&1}\right\rangle.
\]
The matrices $P,Q$ are given by
\[
P= \npmatrix{1&1&1\\0&1&-1};\quad\quad Q = \npmatrix{1&0&-1\\0&1&1\\0&1&-1},
\]
and $D_1=\mathrm{diag}(1,1/2,-1/2)$, $D_2=\mathrm{diag}(0,1/2,1/2)$.
\end{example}

It was shown in \cite{np} that computing the tensor rank of a $3$-tensor over a finite field is NP-complete. However we have the following bound on the tensor rank, which was proved by Kruskal; see \cite[Corollary 1]{kruskal}.

\begin{theorem}\label{thm:trkbound}
	Let $X \in \FN$ be $1$-nondegenerate. Then
	$$\mathrm{trk}(X)\geq \dim_1(X)+\min\{\trk(m_1(u,X)) : u \in \F^k\setminus\{0\} \}-1.$$
\end{theorem}

In the language of rank-metric codes, for  $X \in \FN$ satisfying dim$_1(X)=k$, this inequality is equivalently expressed as
\begin{equation}
\mathrm{trk}(X) \geq k + \drk(\ss_1(X))-1.
\end{equation}


Another operation that is important in the context of tensors, is the so-called {\it contraction} with respect to some indices. Since in this work we will only need one particular contraction, we will not define this concept in general, but only for the following case.

\begin{definition}
Let $k,k',n,m \in \N$, and let 
$$X =\sum_i u_i \otimes v_i \otimes w_i \in \FN, \qquad  Y=\sum_j u_j' \otimes v_j' \otimes w_j' \in \F^{k'\times n \times m}$$ be tensors. We define the {\bf double-dot product} between $X$ and $Y$, as the $2$-tensor (i.e., the matrix) $X:Y \in \Fq^{k \times k'}$ given by
$$X: Y = \sum_{i,j}(v_i\cdot v_j')(w_i \cdot w_j')u_i\otimes u_j'.$$
\end{definition}

In terms of their coordinate tensor representations, if we write $X=(X_{ij\ell})$ and $Y=(Y_{s j \ell})$, then it is straightforward to see that the double-dot product $X:Y$ will have coordinate representation defined by
$$(X:Y)_{is}=\sum_{j,\ell}X_{ij\ell}Y_{s j\ell},  \quad \mbox{ for } 1\leq i \leq k, 1\leq s \leq k'.$$

The definition also extends when one  (or even both) of the tensors is a $2$-tensor, i.e. a matrix, considering $A\in \F^{n\times m}$ as an element in $\F^{1\times n \times m}$. In particular, for two matrices $A, B \in \F^{n\times m}$, we have
$$A:B=\mbox{Tr}(AB^\top),$$
that is, the double-dot product between two matrices corresponds to their trace-product.
Moreover, it is straightforward to prove that, for  $A \in \F^{s\times k}$, $B\in \F^{s\times k'}$,  $X \in \FN$, and $Y\in \F^{k'\times n \times m}$ we have
\begin{gather}
\begin{aligned}\label{eq:colonm1}
m_1(A,X):Y&=A(X:Y), \\
X:m_1(B,Y)&=(X:Y)B^\top.
\end{aligned}
\end{gather}

We now turn to the connection between tensors and rank-metric codes.

\pagecolor{white}
\section{Codes and Tensors}
\label{sec:Tensor}
\label{sec:4}








In the theory of rank-metric codes, a natural representation of the code is with respect to a {\em generating tensor} for it. 
It will become evident that this representation offers greater efficiency in terms of complexity of encoding and storage of the encoder.   

The generator tensor essentially determines an encoding from the information space
$\Fq^k$ to the ambient matrix space $\Fq^{n \times m}$. More specifically, for a given $\Fqk$ code $\C$, an encoder is a an $\fq$-monomorphism 
$$E:\Fq^k :\longrightarrow \Fq^{n \times m}. $$ 
The space of all such encoding maps is {contained in} the space $\mathrm{Hom}_{\Fq}(\Fq^k, \Fq^{n\times m})$, which is an $\Fq$-vector space of dimension $knm$. Moreover, we have that 
$$ \mathrm{Hom}_{\Fq}(\Fq^k, \Fq^{n\times m}) \cong \Fq^{k \times n \times m} $$
as $\Fq$-vector spaces. The isomorphism is explicitly given by:
$$\begin{array}{l}
{\cal E} :\Fq^{k \times n \times m}  \longrightarrow  \mathrm{Hom}_{\Fq}(\Fq^k, \Fq^{n\times m}) 
 : X  \longmapsto  E_X,
\end{array}$$
where 
$$\begin{array}{l}
 E_X: \Fq^k  \longrightarrow \Fq^{n\times m} 
  :g  \longmapsto  m_1(g,X).
\end{array}$$

This yields an analogue of the notion of a generator matrix for rank-metric codes, in the form of a 3-tensor.
\begin{definition}
Let $\C$ be an $\Fqk$ code. A {\bf generator tensor} for code $\C$ is an element $X\in \Fq^{k\times n \times m}$ such that $\ss_1(X)=\C$. 
\end{definition}
Clearly, with respect to this definition, any generator tensor for a code is necessarily $1$-nondegenerate.
The complexity of realizing a code $\C$ as the slice space of a tensor $X$ depends on the tensor rank of $X$ and hence it is of interest to give expressions of generating tensors as minimal sums of simple tensors, and moreover to obtain constructions of codes whose generating tensors have least possible tensor rank.

Let $\mC \subseteq \mat$ be a non-zero code, and let $X_1,X_2$ be generating tensors for $\mC$. By Proposition \ref{prop:trkcharact} we have 
$$\trk(X_1)=\min\{R : \mbox{there are rank 1 matrices } A_1,\ldots,A_R \mbox{ with } \mC \subseteq \langle A_1,\ldots,A_R \rangle\} = \trk(X_2).$$
Therefore the following hold.

\begin{proposition}\label{prop:equaltensor}
 Let $X_1$ and $X_2$ be two generator tensors for the same rank-metric code $\C$. Then 
$$\mathrm{trk}(X_1)=\mathrm{trk}(X_2).$$
Furthermore, if $\mC$ is not the zero code, then this numbers equals the minimum $R>0$ such that~$\mC$ is contained in the span of $R$ rank 1 matrices.
\end{proposition}

\begin{definition}
Let $\C$ be an $\Fqk$ code. The {\bf tensor rank} of $\C$, denoted by $\trk(\C)$, is defined to be the tensor rank of any generator tensor of $\C$.
We say that $\C$ is {{\bf minimum tensor rank}, or {\bf MTR} in short,}  if it meets the bound of Theorem \ref{thm:trkbound}, that is, if 
$$ \trk(\C) = k + \drk(\C) -1.$$
\end{definition}

If $\C$ and $\C'$ are a pair of codes
satisfying $\C'=\varphi(\C)$ for an isometry $\varphi$, then any $R$-base $\A$ for $\C$ yields the $R$-base $\varphi(\A)$ for $\C'$.
Therefore, Proposition \ref{prop:equaltensor} also implies that the tensor rank is invariant under code equivalence.

\begin{proposition}\label{prop:tensequivalent}
 Let $\C,\C'\in \mat$ be equivalent codes. Then $\trk(\C)=\trk(\C')$.
\end{proposition}

%

Let $\C$ be an $\Fqk$ code with generator tensor $X\in \FN$. By the definition of a generator tensor, we have that $\dim_1(X)=\dim_{\Fq}(\C)$. However, $\dim_2(X)$ and $\dim_3(X)$ also have an important role, as explained by the following result.

\begin{proposition}
	Let $\C$ be an $\Fqk$  code with generator tensor $X\in \Fq^{k\times n \times m}$. Then
	$$\dim_2(X)=\dim(\csup(\C)), \qquad \dim_3(X)=\dim(\rsup(\C)).$$
\end{proposition}

\begin{proof}
	Let $A_1,\ldots,A_k$ be a basis of $\C$. Then $Y=\sum_{i=1}^k e_i \otimes A_i$ is a generator tensor for $\C$.
	For any $y \in \Fq^n$, we have that $m_2(y,Y)=\sum_{i=1}^k e_i \otimes (yA_i)=0$
	if and only if $yA_i=0$ for each $i=1,\ldots, k$. This is true if and only if 
	$$y\in \bigcap_{i=1}^k \cs(A_i)^\perp=\left(\sum_{i=k} \cs(A_i)\right)^\perp=\csup(\C)^\perp.$$ 
	In particular $\csup(\C)^\perp$ is the kernel of the map 
	$$m_2(\cdot,Y): \Fq^n \longrightarrow  \Fq^{k\times m}: y \mapsto m_2(y,Y),$$ 
	hence $\dim(\csup(\C))=\dim_2(Y)= \dim \ss_2(\C)=\dim_2(X)$.
\end{proof}

\begin{remark}
 As a consequence, the property of a rank-metric code $\C$ of being nondegenerate can be read from  its generator tensor. Indeed, $\C$ is  nondegenerate if and only if any of its generator tensors is both $2$-nondegenerate and $3$-nondegenerate.
\end{remark}

\begin{remark}
We note that some of the results of this section have been previously considered in the case of $m=n=k$, due to the fact that MRD codes in this situation are in one-to-one correspondence with {\it finite semifields}; that is, nonassociative division algebras. Indeed, tensors and rank-metric codes in this case correspond to algebras which are not necessarily associative. Knuth \cite{Knuth} considered the cubical array of a semifield, which is precisely the co-ordinate tensor introduced in Section 3. The explicit tensor correspondence was outlined in \cite{Liebler} and developed in \cite{Lavrauw}, where the tensor rank was proposed as an interesting invariant of a finite semifield, or equivalently its corresponding slice space.
\end{remark}

\subsection{A Connection with Linear Block Codes}

In \cite{brockett}, the authors draw a connection between tensor rank and linear block codes. See also~\cite[Chapter 18]{algcplex} for an exposition.
This connection provides a lower bound on the tensor rank in terms of the length of a block code hence one can apply coding theoretic bounds to get an estimate for $\trk(X)$.
First, we define the following number from coding theory.
\begin{definition}
	For positive integers $k,d$ we define 
	$$N_q(k,d):= \min \{ N \mid \mbox{there exists an } \fq\text{-}[N,k,d] \text{ code}\}.$$ 
\end{definition}

We will associate a linear block code with a rank-metric codes as follows. 
Let $\C$ be an $\Fqk$ code with tensor rank $R$. 
By Proposition \ref{prop:trkcharact}, we can define the following. 
\begin{definition}
	Let $\C$ be an $\Fqk$ code with $k \ge 1$ and tensor rank $R$.
	A set $\mathcal A=\{A_1,\ldots, A_R\} \subset \F_q^{n \times m}$ of rank 1 matrices 
	such that $\C \subseteq \langle \mathcal A \rangle$ is called an $R${\bf -basis} for $\C$.  
\end{definition} 

Let $\mathcal A=\{A_1,\ldots, A_R\} \subset \F_q^{n \times m}$ be a linearly independent set of matrices of rank 1. 
We define an $\Fq$-linear isomorphism (c.f. \cite[Theorem 18.4]{algcplex}): 
$$\begin{array}{l}
\psi_{\mathcal A}: \langle \mathcal A \rangle  \longrightarrow \Fq^R
:\sum\limits_{i=1}^R \mu_i A_i \longmapsto  \sum\limits_{i=1}^R \mu_i e_i.
\end{array}$$

\begin{definition}
	Let $\C$ be an $\Fqk$ code with tensor rank $R$ and let $\A$ be
	an $R$-basis for $\C$.
	We define the linear block code $C_{\mathcal A}$ to be the image of $\C$ under $\psi_{\mathcal A}$:
	$$C_{\mathcal A}:=\psi_{\mathcal A}(\C).$$
\end{definition}

	For a generator tensor $X=\sum_{r=1}^R u_r \otimes v_r \otimes w_r$ of an $\Fqk$ code $\mC$, any element $M$ of $\C$ can be expressed as
	$$M=m_1(a,X)=\sum_{r=1}^R (a \cdot u_r)(v_r \otimes w_r) $$ 
	for some $a \in \F_q^k$. For $A_r= v_r \otimes w_r$, the image of $M$ element under $\psi_{\mathcal A}$ is $(a \cdot u_r : 1 \leq r \leq R)$. In other words, we have that $C_\A$ is simply the $\F_q$-$[R,k]$ block code with $k \times R$ generator matrix $(u_r : 1 \leq r \leq R)$.

\begin{theorem}[\cite{brockett}]\label{th:brock}
Let $\C$ be an $\F_q$-$[n \times m,k,d]$ code with tensor rank $R$. Let $\mathcal A=\{A_1,...,A_R\}$ be an $R$-basis for $\C$.
Then the following hold.
\begin{enumerate}
\item For every $M \in \C$, $\rk(M)\leq \mathrm{wt}_H(\psi_{\mathcal A}(M))$.
\item $C_{\mathcal A}$ is an $\Fq$-$[R,k,\geq d]$ code. 
\item $\trk(\mC) \geq N_q(k,d).$
\end{enumerate}
\end{theorem}

\begin{proof}
Let $M \in \C$, and let $s=\rk(M)$. Then any expression of $M$ as sum of rank one matrices requires at least $s$ such matrices in the sum. 
In particular, $M=\sum_{i=1}^R \lambda_i A_i$ for some $\lambda_i \in \Fq$ with at least $s$ of the values $\lambda_i$ non-zero. 
Then clearly $s\leq w_H(\psi_\A(M))\leq R$, proving the first statement. The next two statements follow immediately.
\end{proof}


\begin{definition}
	Let $\mC$ be an $\fq$-$[n \times m,k,d]$  code. We say that $\mC$ is {\bf tensor rank extremal} if $\trk(\mC) = N_q(k,d)$.
\end{definition}

Indeed any lower bound on $N_q(k,d)$ provides a lower bound on the tensor rank so the connection to linear block codes can be exploited. In particular, if $\C$ meets the tensor-rank bound, that is, if $R=k+d-1$, then code $C_{\mathcal A}$  is an $\Fq$-$[R,k,R-k+1]$ code and is thus MDS.


For any pair of full-rank matrices $V \in \F^{n \times R}$ and $W \in \F^{m \times R}$, define the $\Fq$-linear map
$$\begin{array}{l}
  \phi_{V,W}: \Fq^R  \longrightarrow  \Fq^{n\times m} : x  \longmapsto  V\diag(x)W^\top,
\end{array}$$
Let $v_r$, $w_r$ denote the $r$th columns of $V$ and $W$, respectively and let $\A=\{ A_r:  1\leq r \leq R\}$, with $A_r=v_r \otimes w_r$ for each $r$. 
Then
\begin{equation}\label{eq:phipsi}
      \phi_{V,W}(\psi_\A (M)) = M,
\end{equation}
for each $M$ in the span of $\A$. This is easy to see, for if $M= \sum_{r=1}^R \lambda_r A_r $ for some $\lambda_r \in \F_q$, then we have
\begin{eqnarray*}
\phi_{V,W}(\psi_\A (M)) &=& \phi_{V,W}\left(\sum_{r=1}^R \lambda_r e_r\right) \\
&=& \phi_{V,W} (\lambda) \\
&=& V\diag(\lambda) W^\top \\
&=& \sum_{r=1}^R \lambda_r v_r \otimes w_r \\
&=& \sum_{r=1}^R \lambda_r A_r = M.
\end{eqnarray*}
This yields the following result.

\begin{lemma}\label{lem:phipsi}
	Let $\C$ be an $\Fqk$ code. Suppose that $\C = V\langle {\mathcal D} \rangle W^\top$ for some set
	${\cal D} = \{D_1,\ldots,D_k\}$ of diagonal matrices and matrices $V\in \F_q^{n \times R}$ and $W \in \F_q^{m \times R}$ of ranks $n,m$ respectively.
	Let $v_r$, $w_r$ denote the $r$th columns of $V$ and $W$, respectively, and define $\A=\{ A_r:  1\leq r \leq R\}$ such that $A_r=v_r \otimes w_r$ for each $r$.
	Then $\phi_{V,W}(\psi_{\mathcal A}(\C)) = \C$.
\end{lemma}

In the next section, we shall be concerned with tensor rank extremal codes (those meeting the bound of Theorem \ref{th:brock}) and in particular with constructions of codes meeting Kruskal's tensor rank bound of (Theorem \ref{thm:trkbound}). One approach will be to view a rank-metric code $\C$ in $\F^{n \times m}$ as the image of an $\Fq$-linear block code under $\phi_{V,W}$.
Then 
$$\phi_{V,W}^{-1}(\C):=\left\{ c \in \Fq^R : V\diag(c)W^\top \in \C \right\},$$
is an $\Fq$-$[R,k]$ code $C$ and in fact we have $C=\psi_\A(\C)$ where $\A=\{ A_r:  1\leq r \leq R\}$ such that $A_r=v_r \otimes w_r$ for each $r$.

We therefore have, using (\ref{eq:phipsi}) and/or Lemma \ref{lem:phipsi}, the following rewriting of Theorem \ref{th:brock}.

\begin{corollary}\label{lem:HammCD}
Let $\C\subseteq \F_q^{n\times m}$ be a rank-metric code of dimension $k$, minimum distance~$d$ and tensor rank $R$. 
Let $\mathcal D = \{D_1,\ldots,D_k\}$ be a $k$-set of  $R \times R$ diagonal matrices and let $\C = V\langle \mathcal D \rangle W^\top$ for matrices $V,W \in \F_q^{n \times R}, \F_q^{m \times R}$ of ranks $n,m$ respectively. The following hold

\begin{enumerate}
\item For every $M \in \C$, $\rk(M)\leq \mathrm{wt}_H(\phi_{V,W}^{-1}(M))$.
\item $\phi_{V,W}^{-1}(\C)$ is an $[R,k,\geq d]$ code. 
\item If $\C$ is tensor rank extremal, then the code $\phi_{V,W}^{-1}(\C)$ is an $\fq$-$[R,k,d]$ code of length $N_q(k,d)$.
In particular, if $\mC$ is MTR then the code $\phi_{V,W}^{-1}(\C)$ is an MDS code.
\end{enumerate}
\end{corollary}

We have also obtained a new proof of the tensor-rank bound.

\begin{corollary}[Tensor-rank bound] \label{coro:trb}
  Let $\C\subseteq \Fq^{n\times m}$ be a rank-metric code. Then
\begin{equation}\label{eq:tensor}\trk(\C) \geq \dim(\C) + \drk(\C)-1.\end{equation}
\end{corollary}

\subsection{Complexity}\label{sec:comp}

We have demonstrated how the encoding map from $\fq^k$ to the ambient space $\fq^{n \times m}$ is represented for a rank metric code by a 3-tensor, namely its generator tensor.
Let $X = \sum_{r=1}^R u_r \otimes v_r \otimes w_r$ be a generator tensor for an $\fq$-$[n\times m,k,d]$ code $\mC$ of tensor rank~$R$.  
The message $a \in \fq^k$ is encoded via $$a \mapsto m_1(a,X) = \sum_{r=1}^R (x\cdot u_r) v_r \otimes w_r = V\diag(xU) W^T,$$ 
where $U=(u_r: 1 \leq r \leq R)$, $V=(v_r: 1 \leq r \leq R)$ and $(w_r: 1 \leq r \leq R)$.
We say that $X$ is in standard form if $U=[I_k|U']$, $V=[I_n |V']$ and $W=[I_m|W']$ for matrices $U',V',W'$ of the required sizes.

$X$ has storage complexity $R(k+n+m)$ in the general case and requires storing up to $R(k+n+m)-k^2-n^2-m^2$ symbols in $\fq$ if $X$ is in standard from. 
Note that the expression of the codeword $c=m_1(a,X)$ as an $\fq$-linear combination of $r$ rank one matrices $v_r\otimes w_r$ is unique and as we outlined before, $\mC$ and $C_{\A}$ are isomorphic. Thus is is sufficient to compute $(c_r)=xU$ in order to represent elements of $\C$, once the generator tensor $X$ is known.
The tensor encoding therefore requires $kR$ multiplications and $(k-1)R$ additions over $\fq$ for arbitrary $U$ of rank $k$ and $k(R-k)$ multiplications and $(k-1)(R-k)$ additions if $U$ is in standard form. 

Of course, being an $\fq$-space, we could also choose to use a generator matrix to represent the encoding map. 
This can be achieved by representing each element
of $\fq^{n \times m}$ as a vector of length $\fq^{nm}$ by the obvious $\fq$-isomorphism  
$$ (M_{ij}) \mapsto (M_{11}\cdots M_{1n}|\cdots |M_{m1}\cdots M_{mn}).$$
Then choose a $k \times nm$ generator matrix $G$ as the encoder. The storage complexity of $G$ is $knm$ and if $G$ is in systematic form it requires $k(nm-k)$ symbols in $\fq$. The encoding complexity of the computation $x \mapsto xG$ for $G$ in standard form then requires $k(nm-k)$ multiplications and $(k-1)(nm-k)$ additions.
We remark that the $k \times nm$ matrix $G$ can simply be obtained from the coordinate tensor representation of $X$ via $G_{it}:=X_{ij\ell}$ where $t=(i-1)j+\ell$ for $1\leq j \leq n$, $1 \leq \ell \leq m$.

We summarize these observations in Table \ref{onlytable}.

\begin{table}[h!]
	\begin{tabular}{|l|c|c|}
		\hline
				& $k\times nm$ Generator Matrix & $k \times n \times m$ Generator Tensor \\
		\hline                         
		Storage                  & $k(mn-k)$                     & $R(k+n+m)-k^2-n^2-m^2$ \\
		\hline
		Encoding Additions       & $(k-1)(nm-k)$                 & $(k-1)(R-k)$  \\
		\hline
		Encoding Multiplications & $k(nm-k)$                     & $k(R-k)$ \\
		\hline
	\end{tabular}
\caption{Complexities}
\label{onlytable}
\end{table}

Since $R\leq nm$, the generator tensor approach in most cases offers complexity lower than that required by the generator matrix encoder. The number of symbols in $\fq$ required to store the standard form generator matrix $G$ exceeds that of the standard generator tensor $X$ if and only if
$$ R < \frac{knm+n^2+m^2}{k+n+m}.$$


We now study codes that meet the tensor-rank bound.

\section{Tensor Rank Extremal Codes}
\label{sec:5}

In this section we consider existence questions on tensor rank extremal and MTR codes. 
Let $k$, $d$, be positive integers. We wish to determine for which $n$, $m \in \mathbb N$ there exists an $\Fqkd$  code $\C$ 
of tensor rank $R=N_q(k,d)$ and in particular, those for which $R=k+d-1$. 

Our approach to this problem is relies on Lemma \ref{lem:HammCD}, which gives a way to obtain tensor rank extremal codes from block codes of minimal length and hence MDS codes from MTR codes. A natural problem is to determine in which cases we can do the converse.

\begin{problem}\label{prob:tro}
Let $n, m$ be positive integers and let $R,k,d$ be positive integers satisfying $R=N_q(k,d)$. 
Find an $\fq$-$[R,k,d]$ code $C$ and a pair of matrices $V \in \Fq^{n \times R}, W \in \Fq^{m \times R}$ such that the code $\phi_{V,W}(C)$ is a rank-metric code of dimension $k$, minimum distance $d$ and tensor rank $R$ (i.e., is a tensor rank extremal code).	
\end{problem}

An interesting special case is given by the following.

\begin{problem}\label{prob:1}
Let $n, m$ be positive integers and let $R,k,d$ be positive integers satisfying $R=k+d-1$. Find an $\fq$-$[R,k]$ MDS code $C$ and a pair of matrices $V \in \Fq^{n \times R}, W \in \Fq^{m \times R}$ such that the code $\phi_{V,W}(C)$ is a rank-metric code of dimension $k$, minimum distance $d$ and tensor rank $R$ (i.e. is an MTR code).
\end{problem}

The answer to these problems clearly depends on $n$ and $m$. Indeed, one immediately observes that both $n$ and $m$ can not be smaller than $d$. Moreover, we can use the Singleton-like bound of Theorem \ref{singbound} to deduce that $n,m$ have to satisfy
$$k\leq \min\{n(m-d+1), m(n-d+1)\}.$$

\begin{definition}
Let $C$ be an $\fq$-$[R,k,d]$ of length $R=N_q(k,d)$. Let $V \in \Fq^{n \times R}$ and let $ W \in \Fq^{m \times R}$. We say that $(C,  V,  W)$ is an \textbf{extremal triple} if it is a solution to Problem \ref{prob:tro}, i.e.  if $\phi_{V,W}(C)$ is a tensor rank extremal code.
\end{definition}

With this notation, given positive integers $k,d$, we wish to determine for which $n,m \in \N$ there exist matrices $V \in \Fq^{n\times R}, W \in \Fq^{m\times R}$ and an $\fq$-$[R=N_q(k,d), k, d]$ code $C$ such that $(C, V, W)$ is an extremal triple. 
It is clear from the definition that this happens if and only if 

\begin{equation}\label{eq:vdw}
\rk(V\diag(c)W^\top) \geq d,
\end{equation}
for every $ c \in C\setminus \{0\}$.

The following result will be useful to us in addressing this problem.
\begin{lemma}\label{lem:rkvdw}
	Let $V\in \fq^{n\times R}, W \in \fq^{m \times R}$ and let $c \in \fq^R$. Let $C_V$ and $C_{W_c}$ denote the row-spaces of $V$ and $W \diag(c)$, respectively. Then
	$$ \rk(V \diag(c)W^T) = \rk(V)  - \dim(C_V^\perp \cap C_{W_c}) = \rk (W)-\dim(C_{W_c} \cap C_V^\perp).$$
\end{lemma}
\begin{proof}
	Suppose first that $V$ and $W$ both have full rank.
	For any $c \in C$, the rank of $V\diag(c)W^\top$ is the rank of the associated bilinear form on 
	$$\varphi:\fq^n\times \fq^m: (x,y) \mapsto xV \diag(c) W^\top y^\top,$$ which is
	$$ n- \dim \ker_L \varphi  = m - \dim \ker_R \varphi.$$ 
	Now $V$ has full rank, and so $\fq^n$ and $C_V$ are isomorphic.
	$$\ker_L\varphi= \{x \in \fq^n : xV\diag(c)W^\top y^\top =0 \:\forall\; y \in \fq^m\} \cong \{ v \in C_V : v\diag(c)W^\top=0 \} = C_V \cap C_{W_c}^\perp.$$
	Similarly, $\ker_R \varphi \cong C_{W_c} \cap C_V^\perp$.
	Now consider the case $\rk V = s \leq n$ and $\rk W = t \leq m$. There exist full rank matrices $A \in \fq^{s \times R}$ and $B \in \fq^{t \times R}$
	such that $AV$ and $BW$ are full rank matrices with the same row-spaces as $V$ and $W$, respectively. Then apply the above argument with $AV$ in place of $V$ and with $BW$ in place of $W$ to complete the proof. 
\end{proof}

We fix some further notation. For an arbitrary matrix $Y \in \fq^{\ell \times R}$ and element $c \in \fq^R$, we write $C_Y$ to denote the row-space of $Y$ and write $C_{Y_c}$ to denote the row-space of $Y\diag(c)$.

It is clear that if for given parameters $k,d$ we have a tensor rank extremal code in $\Fq^{n\times m}$, then we can construct a tensor rank extremal code in a larger ambient space for the same parameters $k,d$. In terms of extremal triples, we can express this observation as follows.

\begin{lemma}\label{lem:biggertriples}
	Let $C$ be an $\fq$-$[R=N_q(k,d), k, d]$ code. Let $V\in \Fq^{n \times R}$ and $W \in \Fq^{m \times R}$ such that $(C,V,W)$ is an extremal triple.  Then for all  integers $n'\geq n$, $m'\geq m$ and for all the matrices $V' \in \Fq^{n'\times R}, W' \in \Fq^{m'\times R}$ such that $\rs(W)\subseteq \rs(V')$ and $\rs(W) \subseteq \rs(W')$, $(C, V',W')$ is an extremal triple.
\end{lemma}

\begin{proof}
	Let $(C, V, W)$ be an extremal triple. Let $V' \in \Fq^{n'\times R}, W' \in \Fq^{m'\times R}$ such that $\rs(W)\subseteq \rs(V')$ and $\rs(W) \subseteq \rs(W')$. Then, there exist $A \in \GL(n',q), B \in \GL(m',q)$ such that 
	$$AV'=\begin{pmatrix} V \\ 0 \end{pmatrix}, \quad BW'=\begin{pmatrix} W \\ 0\end{pmatrix},$$
	Therefore, for every $v \in C \setminus \{0\}$
	$$\rk(V'\diag(v)W'^\top)=\rk(AV'\diag(v)W'^\top B^\top)=\rk\begin{pmatrix} V\diag(v)W^\top& 0 \\ 0 & 0 \end{pmatrix} \geq \rk(V\diag(v)W^\top)= d.$$	
\end{proof}

In particular, this means that in our analysis of Problem \ref{prob:tro} we may assume without loss of generality $V$ and $W$ are full rank matrices.

First we observe that in the case that at least one integer  among $n$ and $m$ is greater or equal than $R$, then it is easy to construct an extremal triple. Suppose that $n\geq R$. 
Let $C$ be a $\fq$-$[R=N_q(k,d), k, d]$ code and let $V,W$ be any full-rank matrices.
By Sylvester's inequality, we get that for every $c\in C\setminus\{0\}$,
$$\rk(V\diag(c)W^\top) \geq \rk(V)+\rk(W\diag(c))-R = \rk(W\diag(c)) $$
For the case $m \geq R$, all columns of $W$ are linearly independent and so $\rk(W\diag(c))=w_H(c) \geq d$.
For the case $R<m$,  $\rk(W\diag(c))\geq \rk(W)-(R-w_H(c)) =w_H(c) \geq d$. In either case the inequality of (\ref{eq:vdw}) is satisfied and clearly holds similarly with the assumption $m \geq R$.
It therefore only remains to consider the case $m,n<R$.

\begin{proposition}\label{prop:equivoptimal}
 Let $C$ be an $\fq$-$[R=N_q(k,d), k, d]$ code. Let $n,m \in \mathbb N$ such that $d\leq n,m<R$ and $V\in\Fq^{n\times R}, W \in \Fq^{m\times R}$. 
\begin{enumerate}
\item \label{pp1} $(C, V, W)$ is an extremal triple.
\item \label{pp2} For every $c \in C\setminus \{0\}$, $\dim(C_V^{\perp} \cap C_{W_c})\leq \rk(W)-d $
\item \label{pp3} For every $c \in C\setminus \{0\}$, $\dim(C_V \cap C_{W_c}^{\perp})\leq \rk(V)-d$
\item \label{pp4} For every $c \in C\setminus \{0\}$, $\dim(C_V^{\perp} + C_{W_c})\geq R-\rk(V)+d$
\item \label{pp5} For every $c \in C\setminus \{0\}$, $\dim(C_V + C_{W_c}^{\perp})\geq R-\rk(W)+d$
\end{enumerate}
\end{proposition}

\begin{proof}
	As we observed before, $\phi_{V,W}(C)$ has tensor rank at most $R$ and dimension $k$ and so is tensor rank extremal if and only if it has minimum rank distance $d$.
	Therefore, from Lemma~\ref{lem:rkvdw}, the equivalence of the first three statements is immediate. 
    The equivalences between (\ref{pp2}) and~(\ref{pp4}) and between (\ref{pp3}) and 
(\ref{pp5}) are a direct consequence of the dimension formula for the sum of two subspaces, which is $\dim(X + Y)+\dim(X \cap Y)= \dim(X)+\dim(Y)$.

The equivalences between (\ref{pp2}) and (\ref{pp5}) and between (\ref{pp3}) and 
(\ref{pp4}) follow from the fact that $(X\cap Y)^\perp=X^\perp + Y^\perp$ and that $\dim(X^\perp)=R-\dim (X)$, for every $X, Y$ subspace of $\Fq^R$.
\end{proof}

\begin{proposition}\label{prop:maxrk+d}
	Let $k,d ,n,m, R$ be positive integers satisfying $ d\leq n,m <R$ and $R=N_q(k,d)$.
	Let $C$ be an $\fq$-$[R, k, d]$ code and let $V\in\Fq^{n\times R}, W \in \Fq^{m\times R}$ such that $C_V$ and $C_W$ are MDS codes.
	If $n+m\geq R + d$, then $(C,V,W)$ is an extremal triple.
\end{proposition}

\begin{proof}
	Let $c \in C\setminus\{0\}$, with $\mathrm{wt}_H(v)=w\geq d$. Since $C_V$ and $C_W$ are MDS codes, we have
	$\rk(V\diag(v))=\min\{n, w\}$ and $ \rk(\diag(v)W^\top)=\min\{m,w\}.$ By the Frobenius rank inequality, we have 
	$$\rk(V\diag(c)W^\top) \geq \rk(V\diag(c))+\rk(\diag(c)W^\top)-\rk(\diag(c)) = \min\{n,w\} + \min\{m,w\} - w.$$ 
	It is easy to check that in all cases, under the assumption that $m+n\geq R+d$, the right hand side of this inequality is at least $ d$. 
\end{proof}
 
We conclude this section with particular construction of an extremal triple involving doubly-extended generalized Reed-Solomon codes 
or Cauchy codes \cite{dur,seroussiroth}, which hence is a partial solution to Problem \ref{prob:1}. Before doing this, we briefly recall some notation. For each $s\in \mathbb N$, let $\Fq[x,y]_{<s}$ denote the $\fq$-space of homogeneous polynomials with degree strictly less than~$s$. Let $\overline{\F}_q=\fq \cup \{\infty\}$ denote the projective line over $\fq$.
For any $$\displaystyle f(x,y)=\sum_{j=0}^{s-1} f_j x^j y^{s-1-j} \in \Fq[x,y]_{<s}$$ we define the map
$$ f:\overline{\F}_q \longrightarrow \fq: \theta \mapsto f(\theta):= \left\{  \begin{array}{cl}
	                       f(\theta,1) & \text{ if } \theta \in \fq, \\
	                       f(1,0)      & \text{ if } \theta = \infty. \\
	\end{array} \right.$$
Let $N\in \mathbb N$. For any $\alpha=(\alpha_1,\ldots, \alpha_N) \in \overline{\F}_q^N$ define the evaluation map
$$\mathrm{ev}_\alpha: \Fq[x,y]_{<s} \longrightarrow \Fq^N : f(x,y) \longmapsto (f(\alpha_1),\ldots, f(\alpha_N)),$$
and similarly for elements of $ \Fq[x]_{<s}$ 
\begin{definition}[see \cite{dur}]
	Let $1 \leq k \leq N-1$ and let $\beta=(\beta_1, \ldots, \beta_N) \in \Fq^N$ and let $\alpha_1,\ldots,\alpha_N$ be pairwise distinct elements of $\overline{\F}_q$.
	The {\bf Cauchy code} $C_k(\alpha,\beta)$ is defined to be the set
	$$C_k(\alpha,\beta):= \left\{ (\beta_1f(\alpha_1), \ldots, \beta_Nf(\alpha_N)) : f \in \Fq[x,y]_{<k} \right\}.$$
\end{definition}

Observe that in our definition we allow the coefficients of $\beta$ to be zero, while the standard definition requires each $\beta \in (\Fq^*)^N$. In particular, the Cauchy code as defined here is MDS if
$\beta \in (\Fq^*)^N$. However, for our purposes, we will sometimes require $\beta$ to have some zero coefficients. 

Let $\ast$ denote the \emph{Schur product} (or {\em Hadamard product}) of two vectors, which is the vector obtained after component-wise multiplication of two vectors of the same length. 
Then we have the expression $$C_k(\alpha,\beta)=\left\{\beta\ast\mathrm{ev}_{\alpha}(f) : f \in \Fq[x,y]_{<k} \right\}.$$ 

The following result gives a construction of MTR codes of dimension $k$ and minimum rank distance $d$, provided that $d<k$. We will see later that in the case $d\geq k$ we can always find a construction of MTR codes for every $m,n \geq d$. Therefore, the case analyzed here is the non-trivial one. 

\begin{theorem}\label{thm:RSconstruction}
Let $0<d<k<R$ be positive integers satisfying $R=k+d-1$ and let $\alpha=(\alpha_1,\ldots, \alpha_R) \in \overline{\F}_q^R$ be a vector such that $\alpha_i \in \Fq$ are pairwise distinct. 
Let $f(x,y) \in \fq[x,y]$ be an irreducible homogeneous polynomial of degree $k$. Let $C=C_k(\alpha,\mathbf{1})$, let $V\in \Fq^{k\times R}$ 
be a parity check matrix of $C_{R-k}(\alpha,\mathrm{ev}_\alpha(f))$ and let $W\in \Fq^{d\times R}$ 
be a generator matrix of $C_{d}(\alpha, \mathbf{1})$. Then $(C, V, W)$ is an extremal triple.
\end{theorem}

\begin{proof}
In order to prove that $(C,V,W)$ is an extremal triple, we use the characterization given in Proposition \ref{prop:equivoptimal}, showing that for every $c\in \C\setminus \{0\}$ we have $\dim(C_V^\perp \cap C_{W_c}) \leq \dim(C_{W_c})-d$. Let $c \in C\setminus\{0\}$. Since $C$ is an MDS code with minimum distance $d$, then $\mathrm{wt}_H(v)\geq d$. Moreover, the code $C_{W_c}$ is obtained from $C_W$ by multiplying the $i$-th coordinate of every codeword by $c_i$, that is, $C_{W_c} = C_d(\alpha,c)$. Since $C_W$ is an MDS code of dimension $d$, we have $\dim(C_{W_c})=d$.  We therefore need to show that
$$C_V^\perp \cap C_{W_c}=\{0\}.$$
Now $c=\mathrm{ev}_\alpha(g)$ for some non-zero $g(x,y) \in \Fq[x,y]_{<k}$, and so $C_{W_v}=C_d(\alpha,\mathrm{ev}_\alpha(g))$. 
Let $b \in C_V^\perp \cap C_{W_c}$. There exist $\mu \in \Fq[x,y]_{<R-k}$, $\lambda \in \Fq[x,y]_{<d}$ such that 
$b=\mathrm{ev}_\alpha(f)\ast \mathrm{ev}_\alpha(\mu)=\mathrm{ev}_\alpha(g)\ast \mathrm{ev}_\alpha(\lambda),$
i.e. 
$$b_i=f(\alpha_i)\mu(\alpha_i)=g(\alpha_i)\lambda(\alpha_i), \quad \mbox{ for } i=1,\ldots, R. $$
From the fact that $\deg f\mu<R$ and $\deg g\lambda <R$, we obtain $f\mu=g\lambda$. Therefore, since $f$ is irreducible, $f$ divides $g$ or $\lambda$. But $\deg g < k$ and $\deg \lambda <d$. This implies $\lambda=0$ and $b=0$.
\end{proof}

\begin{example}
	Let $q=8$, $R=7$, $k=5$ and let $d=R-k+1=3$. Let $\omega$ be a generator of $\F_8^\times$ and let $\alpha=(1,\omega,...,\omega^6)$.
	The polynomial $f(x)=x^5+x^2+1$ is irreducible in $\F_8[x]$. Let $C$ be the $\F_8$-$[7,5,3]$ Reed-Solomon code $C_5(\alpha,{\bf 1})$.
	Let 
	\[
	V = \left[ \begin{array}{ccccccc}
	        1 & 0 & 0 & 0 & 0 & \omega^6 & \omega^2 \\ 
	        0 & 1 & 0 & 0 & 0 & \omega^3 & \omega^5 \\  
	        0 & 0 & 1 & 0 & 0 & \omega^6 & \omega^3 \\  
	        0 & 0 & 0 & 1 & 0 & \omega^5 & \omega^4 \\  
	        0 & 0 & 0 & 0 & 1 & \omega^4 & \omega^2 \\
	           \end{array}     
	    \right], \qquad
	W = \left[ \begin{array}{ccccccc}
	1 & 0 & 0 & \omega^3 & \omega   & 1 & \omega^2 \\ 
	0 & 1 & 0 & \omega^6 & \omega^6 & 1 & \omega^2 \\  
	0 & 0 & 1 & \omega^5 & \omega^4 & 1 & \omega^4 \\  
	\end{array}     
	\right].    
	\] 
	$V$ is a parity check matrix of $C_2(\alpha,ev_\alpha(f))=C_{2}(\alpha,(1,\omega,\omega^2,\omega^4,\omega^4,\omega^2,\omega))$, and 
	$W$ is a generator matrix of $C_3(\alpha,{\bf 1})$.
	It can be checked that for each $c \in C$, we have $C_2(\alpha,ev_\alpha(f)) \cap C_3(\alpha,c) = \{0\}.$ 
	$(C,V,W)$ is an extremal triple and the rank-metric code $\mC=\phi_{V,W}(C)$ is an MTR $\F_8$-$[5 \times 3,5,3]$ code of tensor rank $7$ and is in fact MRD.
\end{example}

\begin{remark}
	The storage complexity cost of the generator tensor for this class of MTR codes is at most $3kd-2k-d$. 
	This bound is exceeded by the bound on the cost of encoding using a generator matrix as described in Section \ref{sec:comp} (which is $k^2(d-1)$) for all $k>d$. 
	The generator tensor encoding cost requires at most $k(d-1)$ multiplications and $(k-1)(d-1)$ additions, while the encoding costs required by a generator matrix
	requires up to $k^2(d-1)$ multiplications and $(k-1)k(d-1)$ additions.
\end{remark}

\begin{remark}
The figure below gives a graphical description of the parameters for which Problem \ref{prob:1} is solved. That is, it represents the parameters for which we do have constructions of MTR codes, the parameters for which we know no MTR codes exist, and the parameters for which the problem is still open. We suppose that $d<k$ are fixed integers, and the axis show increasing $n$ and $m$ (the number of rows and the number of columns of the ambient matrix space). The black hyperbolae represent the Singleton-like bounds, and therefore below them there does not exist any MTR codes. The blue shading represent the construction of MTR codes described in Theorem \ref{thm:RSconstruction} and its transpose. Moreover, by Lemma \ref{lem:biggertriples}, the right-upper quarter-planes having those points as corner points have also a construction of MTR codes. The red line represents the solutions provided by Proposition \ref{prop:maxrk+d}, and again by Lemma \ref{lem:biggertriples}; for each point on it, the right-upper quarter-plane starting from it has solution. As we can see, the area in between the Singleton-like bounds, the red line and the two upper-right quarter-planes starting from the blue dots is not solved yet. For this reason, in the following we will investigate codes which are not necessarily MTR, but have ``small" tensor rank relative to their dimension and minimum distance.

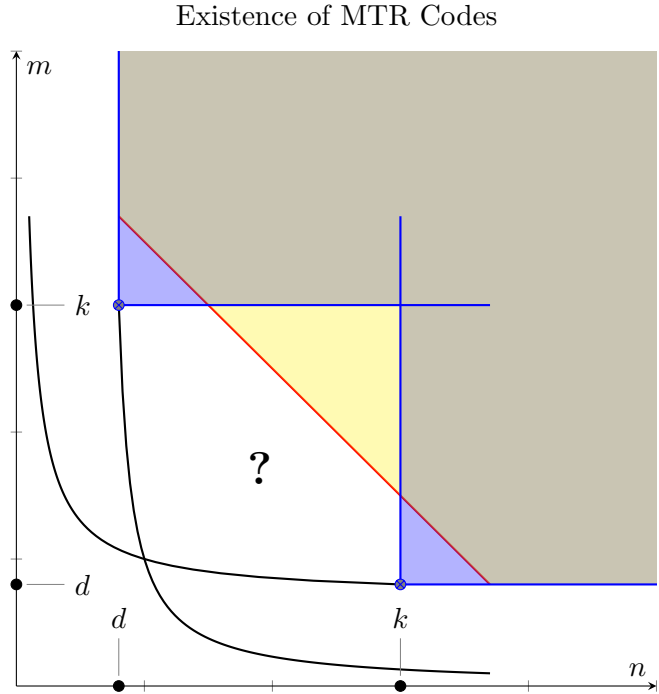
\begin{center}
\begin{tikzpicture}\label{figure}
\begin{axis} [axis lines=middle, 
xlabel=$n$,
ylabel=$m$, 
xmin=0,xmax=50,
ymin=0, ymax=50,
 width=10cm, height=10cm,
title={Existence of MTR Codes},
domain=8:37,
samples=100,
 xticklabels=\empty,
 yticklabels=\empty,
] 
\addplot [thick] (x,{30/(x-7)}); 
\addplot [thick] ({30/(x-7)},x);
\addplot [thick, red] (x,45-x);
\fill [blue,opacity=.3] (80,300) rectangle (700,700);
\fill [blue,opacity=.3] (300,80) rectangle (700,300);
\draw (190,380) node[anchor=north, minimum size=7cm]{\textbf{\LARGE{?}}};
\fill [yellow,opacity=.3] (80,370) -- (370,80)--(700,80)--(700,700)--(80,700);
\addplot [thick, blue] (x, 30);
\addplot [thick, blue] (8, x+22) ;
\addplot [thick, blue] (x+22,8) ;
\addplot [thick, blue] (30,x);
\addplot+[only marks, nodes near coords,
		point meta=explicit symbolic, blue] 
	coordinates {
		(30,8) []
		(8,30) []
	};
\addplot[mark=*]
 coordinates {(0,30)} node[pin=0:{$k$}]{} ;
\addplot[mark=*]
 coordinates {(0,8)} node[pin=0:{$d$}]{} ;
\addplot[mark=*]
 coordinates {(8,0)} node[pin=90:{$d$}]{} ;
\addplot[mark=*]
 coordinates {(30,0)} node[pin=90:{$k$}]{} ;
\end{axis} 
\end{tikzpicture}
\end{center}

\end{remark}
\subsection{Tensor rank of  Delsarte-Gabidulin codes}

In this subsection, we study the tensor rank of Delsarte-Gabidulin codes, which form the best understood family of rank-metric codes. We will give a precise computation of their tensor rank when the dimension of the code over the extension field is $1$, and a non-trivial upper bound when this dimension is strictly greater than $1$. In order to do this, we recall some well-known results on tensors over finite fields. In the literature, the computation of tensor rank of tensors over finite field was mainly studied for complexity purposes. Indeed, the tensor rank of some special tensors reveals the lowest complexity of some operations, such as multiplication between polynomials or between matrices. The interested reader is referred to \cite{algcplex} for a more complete exposition.

First, we need the following result, which ensures that the tensor rank of a vector code is well defined.

\begin{proposition}\label{prop:tensorvectorcode}
Let $C \subseteq \Fm^n$ be a vector code, and let $\Gamma=\{\gamma_1,\ldots, \gamma_m\}$, $\Gamma'=\{\gamma_1',\ldots, \gamma_m'\}$ be two bases of $\F_{q^m}/\F_q$. Then $\trk(\Gamma(C))=\trk(\Gamma'(C))$.
\end{proposition}

\begin{proof}
 By Remark \ref{rem:equivvector}, we have that $\Gamma(C)$ and $\Gamma'(C)$ are equivalent codes in $\mat$, so the result follows by Proposition \ref{prop:tensequivalent}.
\end{proof}

Therefore, by Proposition \ref{prop:tensorvectorcode}, the notion of tensor rank of a vector code is well-defined, and we will denote by $\trk(C)$ the tensor rank of any of its matrix representations.

Let $f\in \Fq[x]$ be a fixed polynomial of degree $k$. The map
$$\F_q[x]_{<m} \times \fq[x]_{<n} \longrightarrow \F_q[x]_{<k}: (g,h) \longmapsto gh \mod f, $$
is clearly bilinear, and so can be represented by a tensor, which we denote by $T_{m,n,k} \in  \Fq^{m\times n \times k}$. 
We have the following result on the tensor rank of $T_{m,n,\ell}$.
\begin{proposition}[\text{\cite[Propositions 14.47, 14.48]{algcplex}}]\label{prop:trkTMNK}
	$T_{m,n,k}$ over $\Fq$ has tensor rank at least $m+ n-1$, and has tensor rank exactly $m+n-1$ if and only
	if $q \geq m+n-2$.
\end{proposition}

\begin{lemma} Let $f\in \Fq[x]$ be an irreducible polynomial of degree $m$, and let $\alpha$ be a root of $f$. 
	Let $C = \langle (1,\alpha, \ldots, \alpha^{m-1}) \rangle_{\F_{q^m}}$ and let $\Gamma=\left\{1,\alpha, \ldots, \alpha^{m-1}\right\}$. 
	The tensor $T_{m,m,m}$ is the generator tensor of the $m$-dimensional code $\Gamma(C)$.
\end{lemma}

\begin{proof}Let $M_f$ denote the companion matrix of the polynomial $f$. Then the map $h \mapsto gh \mod f$  has an associated matrix $g(M_f)$ with
respect to the basis $\{1, x, \ldots, x^{m-1}\}$. Thus,
\begin{equation*}
\ss_1(T_{m,m,m})=\left\{g(M_f) : g \in \F_q[x]_{<m} \right\}=\langle I, M_f, \ldots, M_f^{m-1}\rangle=\Gamma(C). \qedhere
\end{equation*}
\end{proof}

As an immediate corollary, we have a similar statement for the one-dimensional Delsarte-Gabidulin codes in $\Fm^n$.

\begin{corollary}
Let $n\leq m$ be positive integers. The tensor $T_{m,n,m}$ is the generating tensor of a one-dimensional Delsarte-Gabidulin code in $\Fm^n$.
\end{corollary}

\begin{proof}
Denote by $X$ the matrix associated to the map from $\F_q[x]_{<m}$ to $\F_q[x]_{<n}$ defined by 
$$\sum_{i=0}^{m-1}a_ix^i \longmapsto \sum_{i=0}^{n-1}a_ix^i,$$
with respect to the basis $\{1,x,\ldots, x^{m-1}\}$ and $\{1,x,\ldots, x^{n-1}\}$.
Then it is clear by definition that $\ss_1(T_{m,n,m}) = \ss_1(T_{m,m,m})X = \Gamma(C)X$, where $C$ is the  one-dimensional Delsarte-Gabidulin code in $(\Fm)^m$ generated by the vector $(1,\alpha, \ldots, \alpha^{m-1})$, and $\Gamma:=\left\{1,\alpha, \ldots, \alpha^{m-1}\right\}$. Now, for every $v \in C$ and every $i=0,\ldots,m-1$, $\Gamma(\alpha^iv)\in C$. Therefore, for every $\beta \in \Fm, v \in C$, we have $\Gamma(\beta v)X \in \Gamma(C)X$, which means that $\Gamma(C)X$ is equivalent to an $\Fm$-linear code in $\Fm^n$ of dimension $1$. All such codes are Delsarte-Gabidulin codes.
\end{proof}

Using the results above, we give an upper bound on the tensor rank of some special Delsarte-Gabidulin codes.

\begin{proposition}\label{prop:TRKGab} Let $n \leq m$ and let $q \geq m+n-2$. 
	For every $K\leq m$, there exists a $K$-dimensional Delsarte-Gabidulin code in $\Fm^n$ of tensor rank at most $\min\{ mn, K(m+n-1)\}$.

\end{proposition}

\begin{proof}
It is clear that every code in $\Fm^n$ has tensor rank at most $mn$.
Choose as a $K$-dimensional Delsarte-Gabidulin code the code $C\subseteq \Fm^n$ defined as an evaluation code on the points $1,\alpha, \ldots, \alpha^{n-1}$, where $\alpha$ is a primitive element of $\Fm$ over $\Fq$.
Therefore, the code $C$ is the $\Fq$-span of $K$ one-dimensional Delsarte-Gabidulin codes of the form $C_i=\langle (1,\alpha^{q^i}, \alpha^{2q^i}, \ldots, \alpha^{(n-1)q^i})\rangle$. For each $i=1,\ldots, K$, consider the basis $$\Gamma_i=\{1,\alpha^{q^i}, \alpha^{2q^i}, \ldots, \alpha^{(n-1)q^i}\}.$$ Then, $\Gamma_i(C_i)=\ss_1(T_{m,n,m})$, which by Proposition \ref{prop:trkTMNK}, has tensor rank exactly $m+n-1$, and the result follows.
\end{proof}

\begin{remark}
The rank of the tensor $T_{m,m,m}$ for $q<2n-1$ has been studied in connection with the algebraic complexity of multiplication in $\Fm$. This problem remains open in general. We refer to \cite{LavPavZan} for the case $n=3$, and \cite{Ballet} for bounds in the case $q=2$.
\end{remark}

\subsection{Codes with small tensor rank}
In this subsection we give some constructions of codes with tensor rank bounded by above. In order to do that, we rely on the results given in the previous subsection about the tensor rank of Delsarte-Gabidulin codes. Before proceeding with these constructions, we give an auxiliary lemma.

\begin{lemma}\label{lem:inductiveTR}
 Let $\C$ be an $\Fqkd$ code with  tensor rank $R$. Then there exists a subcode $\mathcal D\subset \C$ such that $\dim(\mathcal D)=k-1$, $\drk(\mathcal D)\geq d$ and $\mathrm{trk}(\mathcal D)\leq R-1$.
\end{lemma}

\begin{proof}
 Let $\C$ be a rank-metric code with tensor rank $R$, and let $\A=\{A_1,\ldots,A_R\}$ be an $R$-basis for $\C$. Consider the code $C_{\mathcal A}$. 
 Without loss of generality we may assume that $C_{\mathcal A}$ is systematic in the first $k$ coordinates and so it has a  generator matrix of the form
$$G=(I_k \mid M).$$
Now let $\tilde{ D}$ be the subcode of $C_\A$ generated by all but the first row of $G$. 
The code 
$$\mathcal D:=\phi_{\mathcal A}^{-1}(\tilde{ D})$$
clearly has dimension $k-1$ and minimum distance $\geq d$. Moreover, since $\tilde{ D} \subset \langle e_2,\ldots,e_R\rangle$, we have $\mathcal D \subseteq \langle A_2,\ldots,A_R \rangle$. Therefore
$\trk(\mathcal D)\leq R-1$.
\end{proof}

\begin{proposition}\label{prop:corner}
Let $k,d,n,m$ be positive integers with $d\leq n \leq m$ and let $\rho=\min\{s \in \mathbb N : s(s-d+1)\geq k\}$. If $\rho \leq n$ then there exists an $\Fq$-$[n\times m, k, \ge d]$ code $\C$ such that
$$\trk(\C) \leq k+\min\{\rho(d-1), (\rho-d+1)(\rho-1)\},$$
provided that $q\geq 2\rho-2$.
\end{proposition}

\begin{proof}
Let $K=\rho-d+1$. There exists a Delsarte-Gabidulin code $C \subseteq \F_{q^\rho}^\rho$ of dimension $K$, minimum distance $d$, and by Proposition \ref{prop:TRKGab}, tensor rank at most $\min\{\rho^2, K(2\rho-1)\}=\min\{\rho^2,(\rho-d+1)(2\rho-1)\}$. Let $\Gamma'$ be a basis for $\F_{q^\rho}/\Fq$. Then, the code $\Gamma'(C)$ is an $\Fq$-$[\rho \times \rho, K\rho]$ code and can be embedded in $\Fq^{n\times m}$. Applying Lemma \ref{lem:inductiveTR} $K\rho-k$ times, we get an $\Fq$-$[n\times m, k, \ge d]$ code $\C$ with  $\trk( \C)\leq k+\min\{\rho(d-1),(\rho-d+1)(\rho-1)\}$.
\end{proof}

We now present a refinement of the previous construction, which yields a better upper bound.

\begin{theorem}
	\label{thm:smalltensor}
Let $k,d,n,m$ be positive integers with $d\leq n \leq m$, and $k \leq m(n-d+1)$.  Then there exists an $\Fq$-$[n\times m, k, \ge d]$ code $\C$ such that 
$$\trk(\C) \leq k +\min \left\{\left(\left\lceil\frac{k}{m}\right\rceil+d-1\right)(d-1), \left\lceil \frac{k}{m}\right\rceil\left(\left\lceil\frac{k}{m}\right\rceil+d-2\right)\right\},$$
provided  $q \geq m+\left\lceil\frac{k}{m}\right\rceil+d-3$.
\end{theorem}

\begin{proof}
Let $\mu=\min\{s \in \mathbb N : m(s-d+1)\geq k\}=\lceil \frac{k}{m}\rceil+d-1$. By hypothesis, $\mu \leq n$.  Let $K=\mu-d+1$. There exists a Delsarte-Gabidulin code 
$C \subseteq \F_{q^m}^\mu$ of dimension $K$, minimum distance $d$ and, by Proposition \ref{prop:TRKGab}, tensor rank at most  $\min\{m\mu, K(m+\mu-1)\}=\min\{m \mu,(\mu-d+1)(m+\mu-1)\}$. Let $\Gamma$ be a basis for $\F_{q^m}/\Fq$. Then the code $\Gamma(C)$ is an $\Fq$-$[\mu \times m, Km]$ code, which can be embedded in $\Fq^{n\times m}$. Again, we iteratively apply Lemma~\ref{lem:inductiveTR} $Km-k$ times to get a$\Fq$-$[n\times m, k, \ge d]$ code $\mathcal C$ such that
\begin{align*}
\trk(\mathcal C) &\leq \min\{ m\mu,(\mu-d+1)(m+\mu-1)\}-(\mu-d+1)m +k\\
                          & =k+ \min\{\mu(d-1) ,(\mu-d+1)(\mu-1)\} \\
                          &  =  k +\min \left\{\left(\left\lceil\frac{k}{m}\right\rceil+d-1\right)(d-1), \left\lceil \frac{k}{m}\right\rceil\left(\left\lceil\frac{k}{m}\right\rceil+d-2\right)\right\}. \qedhere
\end{align*} 
\end{proof}

\begin{remark}
	In Proposition \ref{prop:corner}, the essential idea was to take a Delsarte-Gabidulin code whose elements are representable as square matrices, embed it in $\fq^{n \times m}$ and iteratively obtain subcodes with decreasing tensor rank.
	In Theorem \ref{thm:smalltensor} we applied the same principle, but this time chose a Delsarte-Gabidulin code whose elements are representable as rectangular matrices.
	In the first case the initial code is a subspace of $\Fq^{\rho \times \rho}$, while in the second the code is a subspace of $\Fq^{\mu \times m}$.
	The fact that the latter construction gives a smaller upper bound on the tensor rank can be easily verified since $ \left\lceil\frac{k}{m}\right\rceil= \mu-d+1$ and $\mu \leq \rho$.
\end{remark}

\begin{remark}
In fact we stated this result in the most general case, even though we are more interested in the those parameters $k,d,n,m$ that are not covered by the constructions of MTR codes given at the beginning of this section. As a consequence of this result, we get the existence of MTR codes for the same parameters as those arising in Theorem~\ref{thm:RSconstruction}, even though the constructions are quite different.
\end{remark}

\begin{corollary}\label{cor:TOconstruction}
  Let $d,k,n,m$ be positive integers with $d\leq n \leq m$ and $k\leq m$. Then there exists an $\Fqkd$ MTR code $\C$, provided that $q\geq m+d-2$
\end{corollary}

\begin{proof}
 If $k\leq m$, then by Theorem \ref{thm:smalltensor} we get an $\Fqkd$-code $\C$ such that $\trk(\C)\leq k+d-1$ and we deduce the result by the tensor rank bound.
\end{proof}

\begin{remark}
	Observe in Corollary \ref{cor:TOconstruction} we require $q \geq m+d-2$, whereas the construction provided by Theorem \ref{thm:RSconstruction} depends on the existence of a Cauchy code of length $k+d-1$, which we always have for $q\geq k+d-2$.
\end{remark}

\section{Generalized  Ranks of a Code}
\label{sec:6}

 In the sequel, we denote by $\mU$ the set of subspaces of $\F_q^{n \times m}$ that are generated by matrices of rank one. 
 
 \begin{definition}
 Let $\mC$ be an $\Fqk$ code with $k \ge 1$, and let $1 \le r \le k$ be an integer. The \textbf{$r$-th generalized tensor rank} of $C$ is 
$$d_r(\mC)=\min\{\dim(U) : U \in \mU, \ \dim(\mC \cap U) \ge r\}.$$
\end{definition}

It is easy to check that the set of generalized tensor ranks form a code invariant.

\begin{proposition}
Equivalent codes have the same generalized tensor ranks.
\end{proposition}

The next result summarizes the main properties of the generalized  ranks and explains the terminology. It also gives a new proof of the tensor rank bound (Corollary~\ref{coro:trb}).

\begin{theorem}
 Let $\C \subseteq \F_q^{n \times m}$ be a $k$-dimensional code with $k \ge 1$. The following hold.
 \begin{enumerate}[label=(\arabic*)]
 \item $d_1(\C)=\drk(\C)$. \label{1}
 \item $d_k(\C)=\trk(\C)$. \label{2}
 \item For all $1 \le r \le \min\{k,mn-1\}$ we have $d_r(\C) < d_{r+1}(\C)$. \label{3}
 \item For all $1 \le r \le k$ we have $d_r(\C) \ge \drk(\C)+r-1$. In particular,
$\trk(\C) \ge \drk(\C)+k-1$. \label{4}
 \item For all $1 \le r \le k$ we have $d_r(\C) \le \trk(\C)-k+r$.  \label{5}
 \end{enumerate}
\end{theorem}

\begin{proof}
\begin{enumerate}[label=(\arabic*)]
\item Let $M \in \C$ be a matrix with $d=\drk(\C)=\rk(M)$. Write $M=M_1+ \cdots +M_d$, where each $M_i \in \F_q^{n \times m}$ has rank one. Then $U=\langle M_1,\ldots,M_d \rangle$ attains the minimum in the definition of $d_1(\C)$.

\item This follows from Proposition \ref{prop:equaltensor}.

\item Let $U \in \mU$ with $\dim(\C \cap U) \ge r+1$ and $\dim(U)=d_{r+1}(\C)$. Let
$U' \subseteq U$ be a hyperplane of $U$ with $U' \in \mU$. Since $U'+(U \cap \C) \subseteq U$, we have
\begin{eqnarray*}
\dim(U' \cap \C) &=& \dim(U' \cap (U \cap \C)) \\ &=& \dim(U') + \dim(U \cap \C) - \dim(U'+(U \cap \C)) \\
&\ge& \dim(U') +(r+1)-\dim(U) \\ &=&
\dim(U)-1+(r+1)-\dim(U) \\ 
&=& r.
\end{eqnarray*}
By definition, this implies that $d_r(\C) \le \dim(U') = d_{r+1}(\C)-1$.
\item This follows combining \ref{1}, \ref{2}, and \ref{3}.
\item This follows from \ref{2} and \ref{3}. \qedhere
\end{enumerate}
\end{proof}

An interesting application of generalized tensor ranks is the distinction of inequivalent codes, as the following example shows.

\begin{example} \label{exgr1}
Let $q=2$ and $n=m=4$. Let $\mC_1$be the $\mathbb{F}_2$-$[4 \times 4]$ code generated by the matrices
\[
\npmatrix{  1& 0& 0& 0\\ 0& 1& 0& 0\\ 0& 0& 1& 0\\ 0& 0& 0& 1 },
\npmatrix{  0& 1& 0& 0\\ 0& 1& 1& 0\\ 0& 0& 1& 1\\ 1& 0& 1& 0 },
\npmatrix{  0& 0& 1& 0\\ 0& 0& 1& 1\\ 1& 0& 1& 0\\ 1& 1& 0& 1 },
\npmatrix{  0& 0& 0& 1\\ 1& 0& 1& 0\\ 1& 1& 0& 1\\ 0& 1& 0& 1 },
\]
and let $\mC_2$ be the $\mathbb{F}_2$-$[4 \times 4]$ code generated by the matrices
\[
\npmatrix{ 1& 0& 0& 0\\ 0& 1& 0& 0\\ 0& 0& 1& 0\\ 0& 0& 0& 1 },
\npmatrix{ 0& 1& 0& 0\\ 1& 1& 0& 1\\ 0& 0& 1& 1\\ 1& 1& 1& 1 },
\npmatrix{ 0& 0& 1& 0\\ 0& 1& 0& 1\\ 1& 0& 1& 0\\ 0& 1& 0& 0 },
\npmatrix{ 0& 0& 0& 1\\ 0& 1& 1& 1\\ 1& 1& 0& 1\\ 1& 0& 0& 1 }.
\]
It can be checked that both $\mC_1$ and $\mC_2$ are MRD codes of dimension 4, and that their generalized tensor ranks are
$(4,6,8,9)$ and $(4,6,7,9)$, respectively. In particular, $\mC_1$ and $\mC_2$ are not equivalent.
\end{example}

A natural question is whether generalized tensor ranks satisfy a duality property analogous to that of generalized tensor rank weights \cite[Corollary 38]{albgen}. More generally, one may ask if the generalized tensor ranks of a code $\mC$ determine those of the dual code $\mC^\perp$. The answer to this question is negative in general. In the following example, we exhibit two codes that have the same generalized tensor ranks, but whose duals have different generalized tensor ranks.

\begin{example}\label{exgr2}
Let $q=2$ and $n=m=4$. Let $\mC_2$ be the $\mathbb{F}_2$-$[4 \times 4]$ code defined in Example~\ref{exgr1}, and let $\mC_3$ be the $\mathbb{F}_2$-$[4 \times 4]$ code generated by the matrices
\[
\npmatrix{ 1& 0& 0& 0\\ 0& 1& 0& 0\\ 0& 0& 1& 0\\ 0& 0& 0& 1 },
\npmatrix{ 0& 1& 0& 0\\ 0& 1& 1& 1\\ 1& 0& 0& 0\\ 1& 0& 0& 1 },
\npmatrix{ 0& 0& 1& 0\\ 1& 0& 0& 0\\ 1& 0& 0& 1\\ 0& 1& 0& 0 },
\npmatrix{ 0& 0& 0& 1\\ 0& 1& 0& 1\\ 1& 1& 1& 0\\ 0& 1& 1& 0 }.
\]

Then $\mC_2$ and $\mC_3$ have the same dimension, the same minimum distance, and the same generalized tensor ranks, namely, $(4,6,7,9)$. However, the tensor rank of $\mC_2^\perp$ is 14,  and that of $\mC_3^\perp$ is~13.
\end{example}

\section{Operations on Codes}
\label{sec:7}

As we have already seen in Section \ref{sec:Tensor}, the tensor representation of a rank-metric code is the analogue of the generator matrix in the linear block case. We now study the properties of the generator tensor of a code and what information can be read from it. In particular, we focus on operations on rank-metric codes and the corresponding operations on their generator tensors.

Let $M\in \mat$. For any $I\subset [n]$, $J\subseteq [m]$ satisfying $0<|I|<n$, $0<|J|<m$, respectively, we denote by $M_I\in \Fq^{|I|\times m}$ the matrix whose rows are those of $M$ indexed by $I$ and by
$M^J\in \Fq^{n\times |J|}$ the submatrix whose columns are those of $M$ indexed by $J$.

\begin{definition}
Let $\C\subseteq \mat$ be a rank-metric code and $A\in \GL(n,q)$ and let $B \in \GL(m,q)$. 
Let $I\subset [n]$, $J\subseteq [m]$ satisfy $0<|I|<n$, $0<|J|<m$.
We define the 
{\bf row-punctured} and {\bf row-shortened} codes of $\C$ with respect to $A$ and $I$ by
$$\Pi^r(\C, A, I):=\left\{(AM)_{\bar {I}} : M \in \C\right\},\quad\Sigma^r(\C,A,I):=\left\{ (AM)_{\bar{I}} : M\in \C, (AM)_I=0\right\}.$$
 We define the
{\bf column-punctured} and {\bf column-shortened} codes of $\C$ with respect to $B$ and $J$ by
$$\Pi^c(\C, B, J):=\left\{(MB)^{\bar{J}} : M \in \C\right\},\quad\Sigma^c(\C,B,J):=\left\{ (MB)^{\bar{J}} : M\in \C, (MB)^{J}=0\right\}.$$
\end{definition}
Clearly,
\begin{equation}\label{eq:rowcol}
\Pi^r(\C, A, I)=A_{\bar{I}}\C \text{ and } \Pi^c(\C, B, J)= \C B^{\bar{J}}. 
\end{equation}
In particular every row-punctured code of $\C$ has the form $A\C$ for some $\ell \times n$ matrix $A$ of rank~$\ell$ and every column-punctured code
has the form $\C B$ for some $m \times s$ matrix $B$ of rank $s$.

Let $X$ be a generator tensor for $\C \in \fq^{n \times m}$ and let $M \in \C$. Then $M = m_1(\alpha,X)$ for unique $\alpha \in \fq^k$.
Let $I \subset [n]$ and let $A \in \GL(n,q)$. 
Then
$$(AM)_I = A_IM = m_2(A_I,M)=m_2(A_I,m_1(\alpha,X)) = m_1(\alpha,m_2(A_I,X)). $$
In particular, 
\begin{equation}\label{eq:rowm2}
\Pi^r(\C,A,I) = \ss_1(m_2(A_{\bar{I}},X)).
\end{equation} 
Similarly, for any $J \subset [m]$ and $B \in \GL(m,q)$, we have
\begin{equation}\label{eq:colm3}
\Pi^c(\C,B,J) = \ss_1(m_3((B^{\bar{J}})^\top,X)).
\end{equation}
Clearly, $m_2(A_{\bar{I}},X)$ is a generator tensor for $\Pi^r(\C,A,I)$ (respectively $m_3((B^{\bar{J}})^\top,X)$ is a generator tensor for $\Pi^c(\C,B,J)$) if and only if it is $1$-nondegenerate.

\begin{proposition}\label{prop:mindistancepunctured}
	Let $\C$ be an $\Fqk$ code, and let $2\leq d \leq \min\{n,m\}$. The following are equivalent.
	\begin{enumerate}
		\item \label{ppp1} $d(\C) \geq d$.
\item \label{ppp2} For every $A \in \GL(n,q)$ and $I\subseteq [n]$ satisfying $|I|\leq d-1$, the row-punctured code $\Pi^r(\C,A,I)$ has dimension $k$.
\item \label{ppp3} For every $B \in \GL(m,q)$ and $J\subseteq [m]$ satisfying $|J|\leq d-1$, the column-punctured code $\Pi^c(\C,B,J)$ has dimension $k$.
	\end{enumerate} 
\end{proposition}

\begin{proof}
	Let $A \in \GL(n,q)$ and let $I\subseteq [n]$ such that $|I|\leq d-1$.
	The $\fq$-linear epimorphism $$f:\C \longrightarrow \Pi^r(\C,A,I) :M \mapsto (AM)_{\bar{I}},$$
	has non-trivial kernel if and only if $\rk (AM)_{\bar{I}} >0$ for every non-zero $M \in \C$.
	Since $$\rk(AM)_{\bar{I}} \geq \rk(M)-|I| \geq \rk(M)- d + 1 ,$$ $f$ is an isomorphism if and only if $d_R(\C) \geq d$.
	This establishes the equivalence of~(\ref{ppp1}) and (\ref{ppp2}). Similarly, 
(\ref{ppp1}) and (\ref{ppp3}) are equivalent.
\end{proof}

Proposition \ref{prop:mindistancepunctured}, combined with (\ref{eq:rowcol}), (\ref{eq:rowm2}) and (\ref{eq:colm3}) immediately yield the following result.

\begin{corollary}
 Let $X\in \Fq^{k \times n \times m}$ be the generator tensor of an $\Fqk$ code $\C$, and let $2\leq d \leq \min\{n,m\}$. The following are equivalent.
\begin{enumerate}
\item $d(C) \geq d$.
\item For every $A \in \GL(n,q)$ and $I\subseteq [n]$ satisfying $|I|\leq d-1$, $m_2(A_{\bar{I}},X)$ is the generator tensor of $\Pi^r(\C,A,I)$.
\item For every $B \in \GL(m,q)$ and $J\subseteq [m]$ satisfying $|J|\leq d-1$, $m_3((B^{\bar{J}})^\top,X)$ is the generator tensor of $\Pi^c(\C,B,J)$.
\item For every $A \in \fq^{(n-d+1) \times n}$ of rank $n-d+1$, $\dim_1(m_2(A,X))=k$.
\item For every $B \in \fq^{(m-d+1) \times m}$ of rank $m-d+1$, $\dim_1(m_3(B,X))=k$.
\end{enumerate} 
\end{corollary}

\subsection{The Parity Check Tensor}

We define a {\em parity check tensor} of a rank-metric code.

\begin{definition}
 Let $\C$ be an $\Fqk$ code and let $Y \in\fq^{(mn-k)\times n \times m}$. 
 We say that $Y$ is a \textbf{parity check tensor} for $\C$ if 
$$\C=\left\{ M \in \F^{n\times m} \mid Y:M=0 \right\}.$$
\end{definition}

Recall that for a pair of matrices $M=(m_{ij})$ and $N=(n_{ij})$ in $\Fq^{n \times m}$,
$$\langle M, N \rangle=\mathrm{Tr}(AB^\top)=\sum_{i,j} m_{ij}n_{ij}.$$
This operation coincides with the tensor double-dot product when applied to matrices, i.e. $M:N=\langle M, N \rangle$.

\begin{proposition}\label{prop:pcgendual}
Let $Y \in \F^{(mn-k)\times n \times m}$, and let $\C$ be an $\Fqk$ code. Then, $Y$ is a generator tensor for $\C^{\perp}$ if and only if it is a parity check tensor for $\C$.
\end{proposition}

\begin{proof}
	$Y$ is a parity check tensor for $\mC$ if and only if
	$Y:M = 0$ for all $M \in \mC$. From (\ref{eq:colonm1}), this holds if and only if 
	$$0=g(Y:M)=m_1(g,Y):M$$
	for all $g \in \Fq^{nm-k}$ and $M \in \mC$, which holds if and only if $\mC^\perp =\{m_1(g,Y): g \in \Fq^{nm-k}\}$. 
\end{proof}

\begin{corollary}
Let $X \in \F^{k \times n \times m}$ be the generator tensor for an $\Fqk$ code $\C$. A $1$-nondegenerate tensor $Y\in \F^{(mn-k)\times n \times m}$ is a parity check tensor for $\C$ if and only if 
$$X: Y=0.$$
\end{corollary}

Let $\C \subseteq \mat$ be a rank-metric code. We now express the shortened code of a rank-metric code in terms of its parity check tensor. It will be convenient to
use the following duality result.

\begin{theorem}[\text{see \cite[Theorem 3.5]{br17}}]\label{thm:dualityPS}
	Let $\C \subseteq \mat$ be a rank-metric code, $A\in \GL(n,q)$, $B\in\GL(m,q)$. Let $I \subseteq [n]$ with $0<|I|<n$ and let $J \subseteq [m]$ with $0<|J|<m$. 
$$
	\Pi^r(\C,A,I)^\perp  = \Sigma^r(\C^\perp, (A^\top)^{-1}, I), \qquad
	\Pi^c(\C,B,J)^\perp  = \Sigma^c(\C^\perp, (B^\top)^{-1}, J).
$$
\end{theorem}

\begin{corollary}\label{prop:PCshortened}
Let $Y\in \Fq^{(nm-k)\times n \times m}$ be the parity check tensor of an $\Fqk$ code $\C$, let $A\in \GL(n,q)$ and let $B \in \GL(m,q)$. Let $I\subseteq [n]$ and $J \subseteq[m]$. Then 
\begin{enumerate}
\item \label{q1} $\ss_1((m_2(((A^\top)^{-1})_{\bar{I}},Y))=\Sigma^r(\C, A, I)^\perp$.

\item \label{q2} $\ss_1(m_3((B^{-1})^{\bar{J}},Y))=\Sigma^c(\C, (B^\top)^{-1}, J)^\perp$.
\end{enumerate}
\end{corollary}

\begin{proof}
By the duality statement of Theorem \ref{thm:dualityPS}, we have
\begin{equation*}
\label{eq:dualPS}
\Sigma^r(\C,A,I)=\Pi^r(\C^\perp,(A^\top)^{-1},I)^\perp.
\end{equation*}
By Proposition \ref{prop:pcgendual}, $Y$ is a generator tensor for $\C$ and so by Eq. (\ref{eq:rowm2}) we have
$$\ss_1(m_2(((A^\top)^{-1})_{\bar{I}},Y)) = \Pi^r(\C^\perp,(A^\top)^{-1},I),$$ showing that (\ref{q1}) holds. The proof that (\ref{q2}) holds is similar.
\end{proof}

\begin{proposition}\label{prop:shortcriterion}
 Let $\C\subseteq \mat$ be a rank-metric code and let $2\leq d \leq \min\{n,m\}$. The following are equivalent.
\begin{enumerate}
\item \label{qq1} $d(\C)\geq d$.
\item \label{qq2} For every $I\subseteq [n]$ with $|I|=n-d+1$, for every $A\in \GL(n,q)$, $\Sigma^r(\C,A,I)=\{0\}$.
\item  \label{qq3} For every $J\subseteq [m]$ with $|J|=m-d+1$, for every $B\in \GL(m,q)$, $\Sigma^c(\C,B,J)=\{0\}$.
\end{enumerate}
\end{proposition}

\begin{proof}
From Proposition \ref{prop:mindistancepunctured}, $d_R(\C)\geq d$ if and only if $\Pi^r(\C,A,\bar{I})$ and $\C$ are isomorphic under the map
$:M\mapsto (AM)_{I}.$ The kernel of this map is $\Sigma^r(\C,A,I)$, which shows the equivalence of (\ref{qq1}) and (\ref{qq2}). The equivalence of (\ref{qq1}) and (\ref{qq3}) follows similarly.  
\end{proof}

As a direct consequence of Corollary \ref{prop:PCshortened} and Proposition \ref{prop:shortcriterion}, we get the following result that relates the minimum distance of a rank-metric code with any of its parity check tensors.

\begin{corollary}
Let $Y\in \Fq^{(nm-k)\times n \times m}$ be the parity check tensor of an $\Fqk$ code $\C$, and let $2\leq d \leq \min\{n,m\}$. The following are equivalent.
\begin{enumerate}
\item $d(\C)\geq d$.
\item for every $A\in\Fq^{(d-1)\times n}$ of full rank, $\ss_1(m_2(A,Y))=\Fq^{(d-1)\times m}$.
\item for every $B\in \Fq^{(d-1)\times m}$ of full rank, $\ss_1(m_3(B,Y))=\Fq^{n\times (d-1)}$.
\end{enumerate}
\end{corollary}

\section*{Acknowledgement}
Alessandro Neri would like to thank the hospitality of University College Dublin during his 6-month stay there, supported by the Swiss National Science Foundation mobility Grant n. 169510/2.

\end{document}